\documentclass[conference,10pt]{IEEEtran}
\usepackage{algorithm,multirow}
\usepackage[noend]{algorithmic}
\usepackage{graphicx}
\usepackage{comment}
\usepackage{cite}
\usepackage{multirow}
\usepackage{amssymb,amsthm,mathrsfs}
\usepackage{amsfonts,epsf,graphicx,times,amsmath,multirow,url,cite}
\usepackage{float}
\usepackage{flexisym}
\usepackage[capitalise]{cleveref}

\newtheorem{lemma}{Lemma}
\usepackage{amsthm}
\usepackage{subcaption}
\usepackage{color}

\graphicspath{{figs/}}
	\DeclareGraphicsExtensions{.eps,.pdf,.jpeg,.png}

\usepackage{array}
\newcolumntype{L}[1]{>{\raggedright\let\newline\\\arraybackslash\hspace{0pt}}m{#1}}
\newcolumntype{C}[1]{>{\centering\let\newline\\\arraybackslash\hspace{0pt}}m{#1}}
\newcolumntype{R}[1]{>{\raggedleft\let\newline\\\arraybackslash\hspace{0pt}}m{#1}}

\makeatletter
\def\ps@headings{%
\def\@oddhead{\mbox{}\scriptsize\rightmark \hfil \thepage}%
\def\@evenhead{\scriptsize\thepage \hfil \leftmark\mbox{}}%
\def\@oddfoot{}%
\def\@evenfoot{}}
\makeatother
\newcommand{\sol}{\texttt{DRAPS}}

\begin{document}

\date{}


\title{\huge DRAPS: Dynamic and Resource-Aware Placement Scheme for Docker Containers in a Heterogeneous Cluster}

\author{
\IEEEauthorblockN{
Ying Mao\IEEEauthorrefmark{1},
Jenna Oak\IEEEauthorrefmark{1},
Anthony Pompili\IEEEauthorrefmark{1},
Daniel Beer\IEEEauthorrefmark{1},
Tao Han\IEEEauthorrefmark{2} 
Peizhao Hu\IEEEauthorrefmark{3}
}
\IEEEauthorblockA{\IEEEauthorrefmark{1}Department of Computer Science,
The College of New Jersey, Email: \{maoy, oakj1, pompila1, beerd1\}@tcnj.edu}
\IEEEauthorblockA{\IEEEauthorrefmark{2}Department of Electrical \& Computer Engineering,
University of North Carolina at Charlotte, Email: tao.han@uncc.edu}
\IEEEauthorblockA{\IEEEauthorrefmark{3}Department of Computer Science,
Rochester Institute of Technology, Email: ph@cs.rit.edu}
}

\maketitle

\begin{abstract}
Virtualization is a promising technology that has
facilitated cloud computing to become the next wave of the Internet revolution.
Adopted by data centers, millions of applications
that are powered by various virtual machines improve the quality
of services. Although virtual machines are well-isolated among
each other, they suffer from redundant boot volumes and slow
provisioning time. To address limitations, containers were born
to deploy and run distributed applications without launching
entire virtual machines. As a dominant player, Docker is an open-
source implementation of container technology. When managing
a cluster of Docker containers, the management tool, Swarmkit,
does not take the heterogeneities in both physical nodes and
virtualized containers into consideration. The heterogeneity lies
in the fact that different nodes in the cluster may have various
configurations, concerning resource types and availabilities, etc., and the demands
generated by services are varied, such as CPU-intensive (e.g.
Clustering services) as well as memory-intensive (e.g. Web services). In this
paper, we target on investigating the Docker container cluster
and developed, DRAPS, a resource-aware placement scheme to
boost the system performance in a heterogeneous cluster.
\end{abstract}

\section{Introduction}
In the past few decades, we have witnessed a spectacular information explosion over the Internet. 
Hundreds of thousands of users are consuming the Internet through various services, such as websites, mobile applications, and online games.
The service providers, at the back-end side, are supported by state-of-the-art infrastructures on the cloud, such as Amazon Web Service~\cite{aws} and
Microsoft Azure~\cite{azure}. 
Focusing on providing the services at scale, 
virtualization is one of the emerging technologies used in data centers and cloud environments to 
improve both hardware and development efficiency.

At the system level, the virtual machine is a widely-adopted virtualization method~\cite{vm}, which isolates CPU, memory, 
block I/O, network resources, etc~\cite{jithin2014virtual}.
In a large-scale system, however,
providing services through virtual machines would mean that the users are probably 
running many duplicate instances of the same OS and many redundant boot volumes~\cite{medina2014survey}. 
Recent research shows that virtual machines suffer from noticeable performance overhead, large storage requirement, and 
limited scalability~\cite{xu2014managing}. 

To address these limitations, containers were designed for deploying
and running distributed applications without launching entire virtual machines. 
Instead, multiple isolated service units of the application, 
called containers, share the host operating system and physical resources.
The concept of container virtualization is yesterday's news;
Unix-like operating systems leveraged the technology for over a decade and modern 
big data processing plforms utilize containers as a basic computing unit~\cite{wang2014fresh, wang2015omo, wang2017seina}. 
However, new containerization platforms, such as Docker, make it into the mainstream of application development.
Based on previously available open-source technologies 
(e.g. cgroup), Docker introduces a way of simplifying the tooling required to create and manage containers.
On a physical machine, containers are essentially just regular processes; in the system view, that enjoy a virtualized resource environment, not only just CPU and
memory, but also bandwidth, ports, disk i/o, etc. 

We use ``Docker run image'' command to start a Docker container on physical machines. 
In addition to the disk image that we would like to
initiate, users can specify a few options, such as ``-m'' and ``-c'', to limit a container's access to resources. 
While options set a maximum amount, resource contention still happens among containers on every host machine.
Upon receiving ``Docker run'' commands from clients, the cluster, as the first step, should select a 
physical machine to host those containers. The default container placement scheme, named Spread, uses a bin-pack strategy and tries to
assign a container on the node with the fewest running containers. 
While Spread aims to equally distribute tasks among all nodes, it omits two major characteristics of the system. First of all,
the nodes in a cluster do not necessary have to be identical with each other. It is a common setting to have multiple node types, 
in terms of total resource, in the cluster. For example, a cutting edge server can easily run more processes concurrently than a off-the-shelf desktop.
Secondly, the resource demands from containers are different. 
Starting with various images, services provided by containers are varied, which leads to a diverse resource demands. 
For instance, a clustering service, e.g. Kmeans, may need more computational power and a logging service, e.g. Logstash, 
may request more bandwidth. 

In this project, we propose a new container placement scheme, \sol, a Dynamic
and Resource-Aware Placement Scheme. Different from the default Spread scheme, \sol~ assigns containers based on current available 
resources in a heterogeneous cluster and dynamic demands from containers of various services. First, \sol~ identifies
the dominant resource type of a service by monitoring containers that offer this service. It, then, places the containers with 
complementary needs to the same machine in order to reduce the balance resource usages on the nodes. 
If one type of resource, finally, becomes a bottleneck in the system, it migrates the resource-intensive containers to other nodes.
Our main contributions are as follows:
\begin{itemize}
 \item We introduce the concept of dominant resource type that considers the dynamic demands from different services.
 \item We propose a complete container placement scheme, \sol, which assigns the tasks to appropriate nodes and balance resource usages
 in a heterogeneous cluster.
 \item We implement \sol~ into the popular container orchestration tool, Swarmkit, and conduct the experiment with 18 services in 4 types.
 The evaluation demonstrates that \sol~ outperforms the default Spread and reduces usage as much as 42.6\% on one specific node.
 
\end{itemize}

\section{Related Work}
Virtualization serves as one of the fundamental technologies in cloud computing systems. 
As a popular application, virtual machines (VMs) have been studied for decades.
However, in the reality, VMs suffer from noticeable
performance overhead, large storage requirement, and limited scalability~\cite{xu2014managing}.
More recently, containerization, a lightweight virtualization technique,  
is drawing increasing popularity from different aspects and on different 
platfroms~\cite{secpod, men2012interface, 
cheng2014efficiently, cheng2013qbdj, cheng2016efficient, cheng2014efficient, edos, bhimani2016, bhimani2017,tang2012gpu, du2015gpu, zhao2012mesh}.

The benefits and challenges of containerized systems have been studied in many aspects.
A comprehensive performance study is presented in ~\cite{felter2015updated}, where it explores the traditional
virtual machine deployments, and contrast them with the use
of Linux containers. The evaluation focuses on overheads and experiments 
that show containers' resulting performance to be equal or superior to VMs performances.
Although containers outperform VMs, 
the research~\cite{slacker} shows that the startup latency is considerably larger than expected.
This is due to a layered and distributed image architecture, in which copying package data accounts
for most of container startup time. The authors propose Slacker which can significantly reduce the startup latency.
While Slacker reduces the amount of copying and transferring packages, if the image is locally available, the
startup could be even faster. CoMICon~\cite{nathan2017comicon} addresses the problem by sharing the image in a cooperative manner. From different aspect, SCoPe~\cite{scope} tries to manage the provisioning time for large scale containers. 
It presents a statistical model, used to guide provisioning strategy, to characterize the provisioning time in terms of system features.

Besides the investigations on standalone containers, the cluster of containers is another important aspect in this field.
Docker Swarmkit~\cite{swarmkit} and Google Kubernetes~\cite{bernstein2014containers} are dominant cluster management tools in the market.
The authors of ~\cite{gog2016firmament}, first, conduct a comparison study of scalabilities under both of them. Then, firmament is proposed
to achieve low latency in large-scale clusters by using multiple min-cost max-flow algorithms. 
On the other hand, focusing on workload scheduling, the paper~\cite{kaewkasi2017improvement} describes an Ant Colony Optimization algorithm
for a cluster of Docker containers. However, the algorithm does not distinguish various containers, which usually have a divese requirements.

In this paper, we investigate the container orchestration in the prospective of resource awareness. 
While users can set limits on resources, containers are still competing for resources in a physical machine.
Starting from different images, the containers target various services, which results in different requirements on resources.
Through analyzing the dynamic resource demands, our work studies a node placement scheme that balance the resource usages in a 
heterogeneous cluster. 


\section{Background and Motivation}
\label{back}
\subsection{Docker Containers}

A Docker worker machine runs a local Docker daemon. 
New containers may be created on a worker by sending commands to its local daemon, such as ``docker
run -it ubuntu bash''.
A Docker container image is a lightweight, stand-alone, executable package of a piece of software that 
includes everything needed to run it: code, run-time, system tools, system libraries, and settings.
In general, each container targets a specific service of an application. If the application needs to scale up this particular service, 
it initiates duplicated containers by using the same image. One physical machine can host many applications with various services in a standalone mode.

\subsection{Container Orchestration}
When deploying applications into a production environment, it's difficult to achieve resilience and scalability on a single container host. 
Typically, a multi-node cluster is used to provide the infrastructures for running containers at scale. Introduced by Docker, SwarmKit is an open source toolkit
for container orchestration in the cluster environment. 

There are two types of nodes in a cluster that are running SwarmKit, worker nodes, and manager nodes.
Worker nodes are responsible for running tasks; on the other hand, 
manager nodes accept specifications from the user and are responsible for reconciling the desired state with the actual cluster state.

A Docker container can be initiated with specific requirements (e.g. memory and CPU) and user-defined labels.
The scheduler that runs on a manager combines the user-input information with states of each node to make various scheduling decisions, 
such as choosing the best node to perform a task. 
Specifically, it utilizes filters and scheduling strategies to assign tasks. 
There are four filters available.
\emph{ReadyFilter:} checks that the node is ready to schedule tasks;
\emph{ResourceFilter:} checks that the node has enough resources available to run;
\emph{PluginFilter:} checks that the node has a specific volume plugin installed.
\emph{ConstraintFilter:} selects only nodes that match certain labels.
If there are multiple nodes that pass the filtering process, 
SwarmKit supports three scheduling strategies: spread (currently available), binpack, and random (under development based on Swarm Mode).
\emph{Spread strategy:} places a container on the node with the fewest running containers.
\emph{Binpack strategy:} places a container onto the most packed node in the cluster.
\emph{Random strategy:} randomly places the container into the cluster.

\begin{figure}[ht]
\centering
\includegraphics[width=0.8\linewidth]{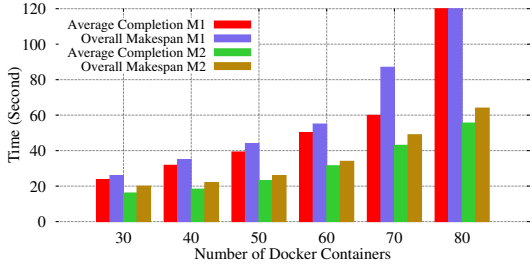}
\caption{Starting Dockers on a single machine}
\label{fig:start-single} 
\end{figure}

The default spread strategy, which 
attempts to schedule a service task based on the number of active containers on each node, 
can roughly assess the resources on the nodes. 
However this assessment fails to reflect various nodes in a heterogeneous cluster setting.
Considering the heterogeneity, the nodes in such a cluster have different configurations in terms of memory, CPU, and network.
Therefore, running the same amount of containers on these nodes results in different experiences.
Fig~\ref{fig:start-single} plots the average starting delay of and overall makespan of the set of Tomcat Docker containers. 
We conduct the experiments on two machines, M1 with 8GB memory, 4-core CPU and M2 has
16GB memory and 8-core CPU. On each particular machine, M1 or M2, we can see that the more containers it hosts, the larger the starting delay 
and makespan. However, M1 costs 23.67s on average to start 30 Tomcat containers and M2 costs 18.32s to start 40 containers. 
Additionally, when trying to initiate 80 Tomcat containers, M1 fails to complete the job and M2 finishes it.

\section{\sol~ System}
\label{sys}
\subsection{Framework of Manager and Worker Nodes}
As described in the previous section, there are multiple managers and workers in the system.
A manager has six hierarchical modules.
{\bf Client API} accepts the commands from clients and creates service objects.
{\bf Orchestrator} handles the lifecycle of service objects and manages mechanics for service discovery and load balancing.
{\bf Allocator} provides network model specific allocation functionality and allocates IP addresses to tasks.
{\bf Scheduler} assigns tasks to worker nodes.
{\bf Dispatcher} communicates with worker nodes, checks their states, and collects the {\bf heartbeats} from them.

A worker node, on the other hand, manages the Dockers containers and sends back their states to managers through periodical heartbeat messages.
An {\bf executor} is used to run the tasks that are assigned to the containers in this worker.

\subsection{\sol~ modules}
To simplify the implementation, we integrate the \sol~ components into the current framework. As shown on Fig~\ref{fig:drap},
it mainly consists of three parts: a container monitor that resides in the worker nodes, a worker monitor, and a \sol~ scheduler that
implement in manager nodes.

{\bf Container Monitor}: a container monitor collects the runtime resources usage statistics of Docker containers on worker nodes. 
At each application level, the monitored resources contain memory, CPU percentage, block I/O, and network I/O. 
The average usage report in a given time window of top users will be injected into the DRAP-Heartbeat messages and sent back to managers.
At the host system level, the tracking information includes I/O wait, reminder percentage of available memory, CPU, and bandwidth. 
The information is used by worker nodes to conduct a self-examination to identify its own bottleneck. 
If a bottleneck is found, a DRAP-Alert message will be produced and sent back to managers.

{\bf Work Monitor}: a worker monitor processes the messages from worker nodes. It maintains a table for each worker and the corresponding containers. 
Through analyzing the data, it will generate tasks, such as migrating a resource-intensive container to another host.

{\bf DRAP-Scheduler}: the DRAP-Scheduler assigns a task to a specific node based on the current available resources. For a duplicated Docker container,
DRAP-Scheduler checks its characteristics on resource consumption, 
such as memory intensity, through the records of the previous containers in the same services.

\begin{figure}[ht]
\centering
\includegraphics[width=0.8\linewidth]{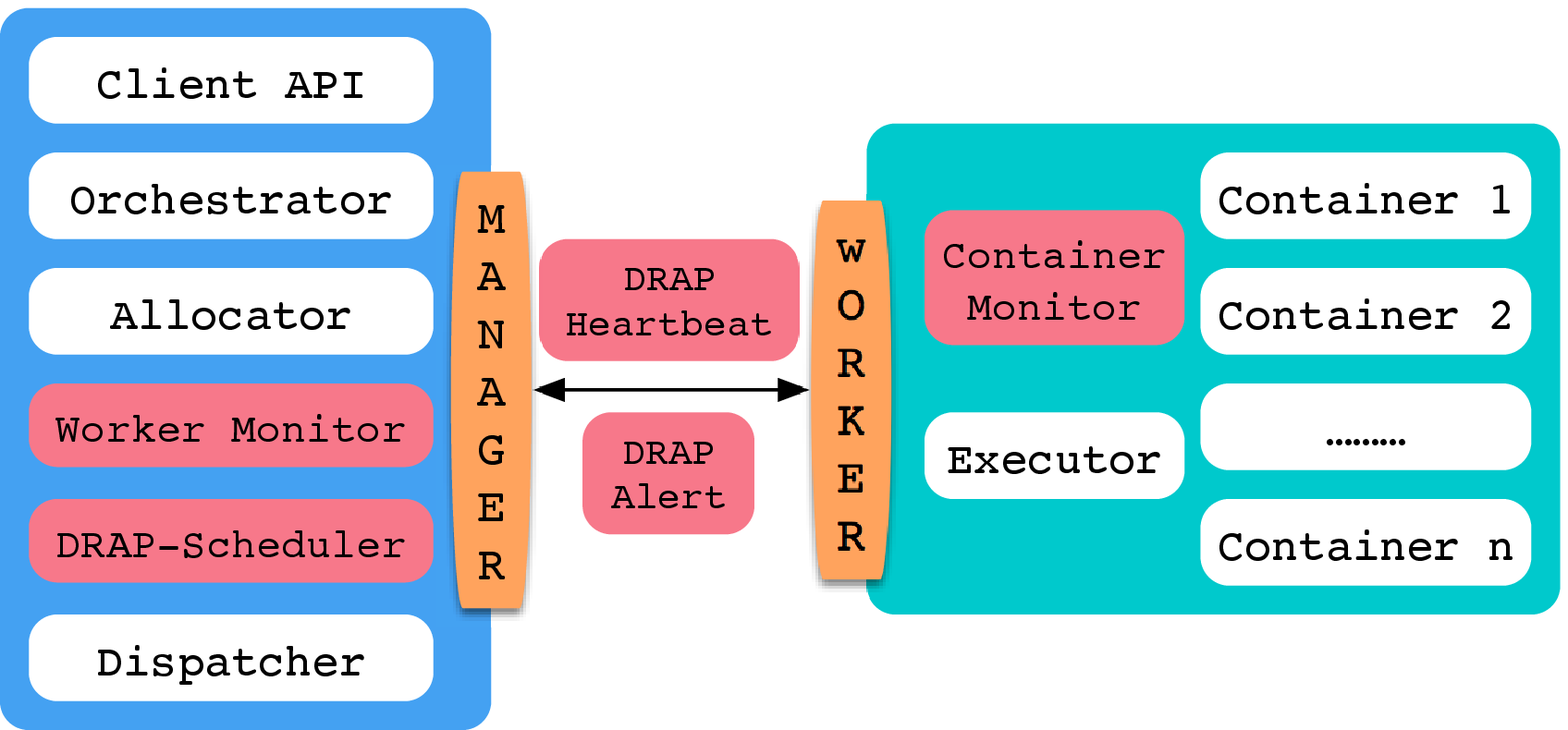}
\caption{Docker Framework with \sol~ Implemention}
\label{fig:drap} 
\end{figure}
\section{Problem Formulation}
The \sol~ scheduler aims to optimize the container placement such that the available
resources on each worker node are maximized. In this paper, we assume that a container requires multiple resources such as memory, CPU, bandwidth, and I/O for running its services. Since the services and their workloads in a container change over time, the resource requirements in a container also exhibit temporal
dynamics. Therefore, we formulate the resource requirements of a container as a function of time.
Denote $r_i^{k}(t)$ as the $k$th resource requirement of the $i$th container at time $t$. Let $x_{i,j}=\{0, 1\}$ be the container placement indicator.
If $x_{i,j}=1$, the $i$th container is placed in the $j$th work node. Denote $W_j^k$ as the total amount of the $k$th
resource in the $j$th work node. Let $\mathcal{C}$, $\mathcal{N}$, $\mathcal{K}$ be the set of containers, work nodes,
and the resources, respectively. The utilization ratio of the $k$ resource in the $j$th work node can be expressed as
\begin{equation}\label{eq:work_util_ratio}
  u^k_j(t)=\frac{\sum_{i\in\mathcal{C}}x_{i,j}r_i^{k}(t)}{W_j^k}
\end{equation}
We assume that the utilization ratio of the $j$th work node is defined by its highest utilized resource.
Then, the utilization ratio of the $j$th work node is $\max_{k\in\mathcal{K}}u^k_j(t)$.
The highest resource utilization among all the work nodes can be identified as
\begin{equation}\label{eq:util_worker}
  \nu=\max_{j\in\mathcal{N}}\max_{k\in\mathcal{K}}u^k_j(t).
\end{equation}
Since our objective when designing the \sol~ scheduler is to maximize the available
resources in each worker node, the \sol~ scheduling problem can be formulated as
\begin{align}\label{eq:sched_problem}
  \max_{x_{i,j}} &\;\;\;\;\; \nu \\
  s.t. &\;\;\;\;\; \sum_j x_{i,j}=1;\forall i\in\mathcal{C};\\
       &\;\;\;\;\; u^k_j(t)\leq 1, \forall k\in\mathcal{K}, \forall j\in\mathcal{N}.
\end{align}
The constraint in E.q. (4) requires that each container should be placed in
one worker node. The constrain in E.q. (5) enforces that the utilization ratio of any resource in a worker is less than one.
\begin{lemma}
The \sol~ scheduling problem is NP-hardness.
\end{lemma}
\begin{proof}
In proving the Lemma, we consider a simple case of the \sol~ scheduling problem in which the resource
requirements of each container are constant over time. The simplified \sol~ scheduling problem equals to
the multidimensional bin packing problem which is NP-hard~\cite{Bansal:2006:IAA,Meyerson:2013:OMLB,Im:2015:TBO}.
Hence, the lemma can be proved by reducing any instance of the multidimensional bin packing to the simplified \sol~
scheduling problem. For the sake of simplicity, we omit the detail proof in the paper.
\end{proof}

\section{\sol~ in a Heterogeneous Cluster}
Previously, we discussed the different modules in \sol~ and their major responsibilities. We also formulated the \sol~ scheduling problem and proved that the problem is NP-hard.
In this section, we present the detailed design of \sol~with heuristic container placement and migration algorithms, in a heterogeneous cluster,
which aims to increase resource availability on each worker node and boost the service performance by  approximating the optimal solution of the \sol~ scheduling problem.
To achieve the objectives, \sol~ system consists of three strategies:
1) Identify dominant resource demands of containers;
2) Initial container placement;
3) Migrate a container

\begin{figure*}[ht]
   \centering
      \begin{subfigure}[t]{0.24\linewidth}
\centering
      \includegraphics[width=\linewidth]{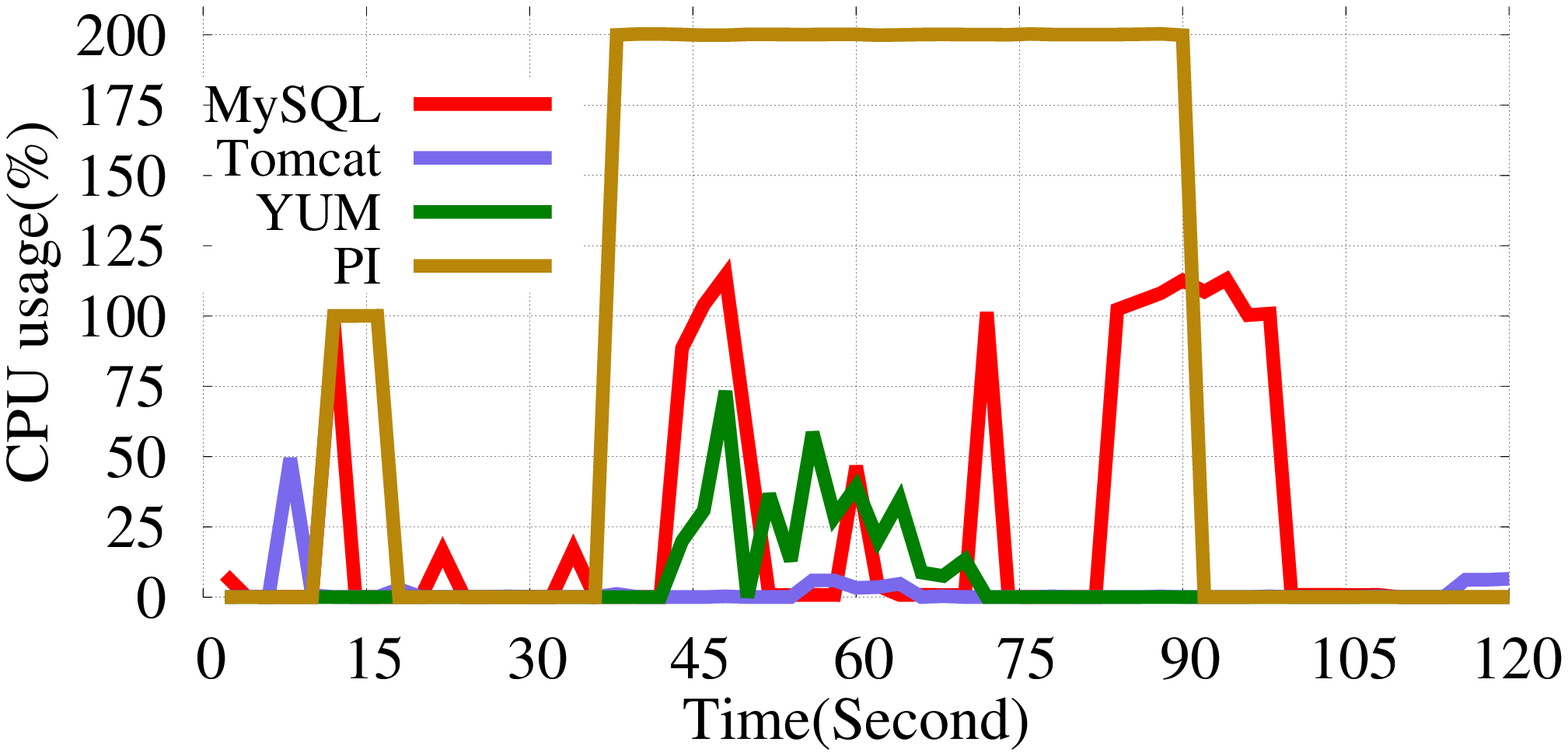}
      \vspace{-0.15in}
      \caption{Usage of CPU}
      \label{fig:cpu}
      \end{subfigure} %
      \begin{subfigure}[t]{0.24\linewidth}
\centering
      \includegraphics[width=\linewidth]{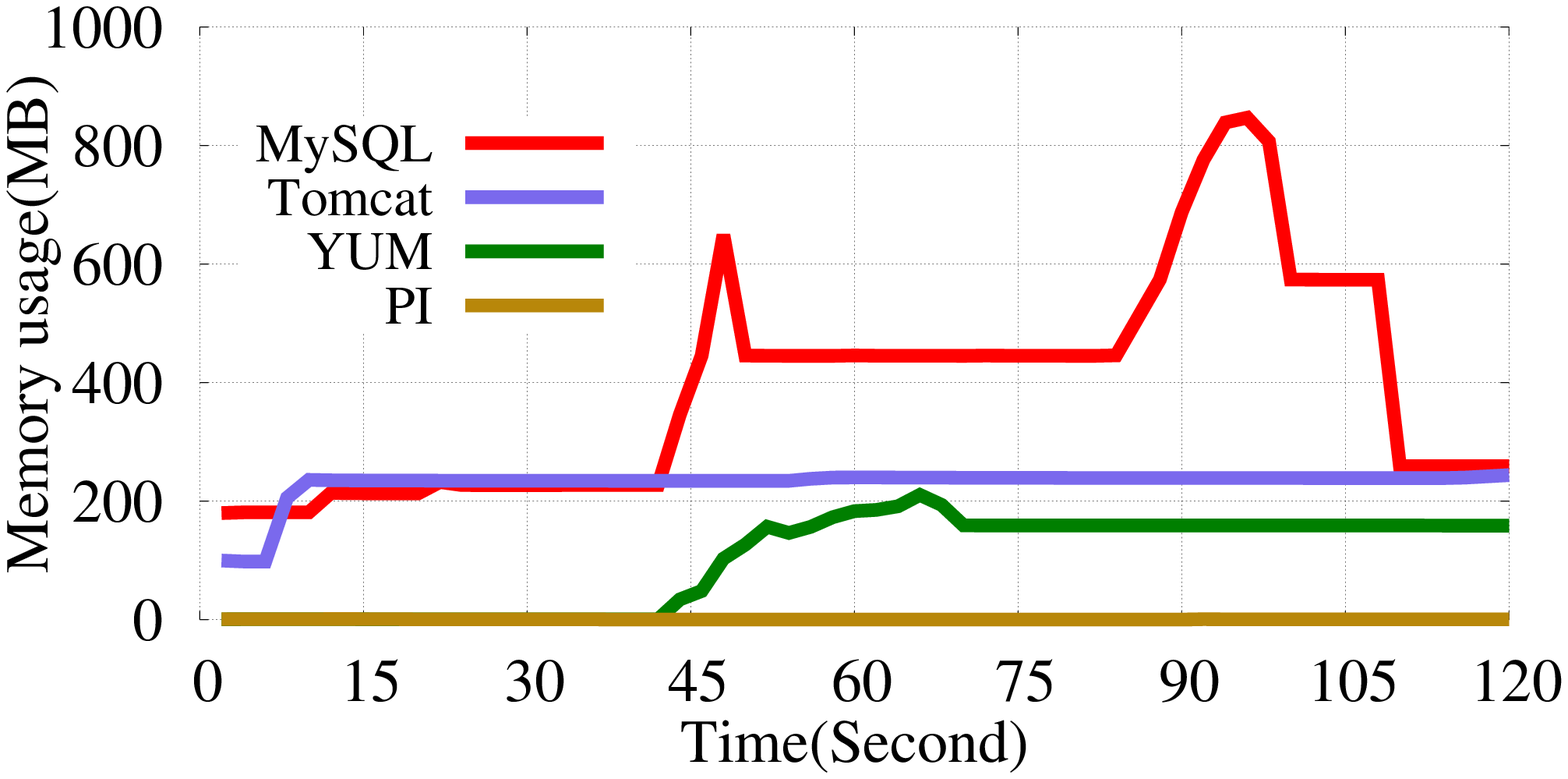}
      \vspace{-0.15in}
      \caption{Usage of Memory}
      \label{fig:mem}
      \end{subfigure} %
      \begin{subfigure}[t]{0.24\linewidth}
\centering
      \includegraphics[width=\linewidth]{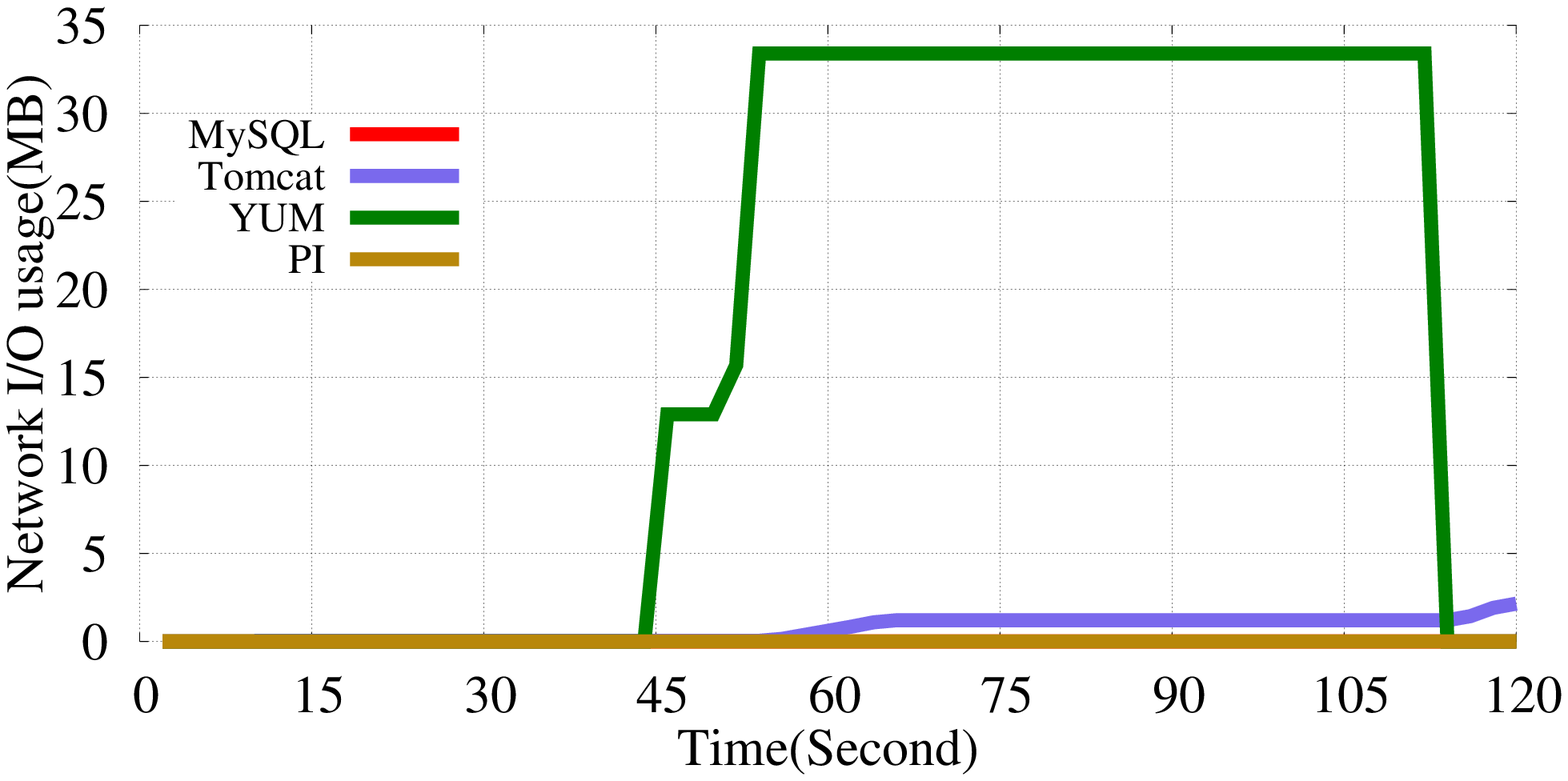}
      \vspace{-0.15in}
      \caption{Usage of Network I/O}
      \label{fig:netio}
      \end{subfigure} %
      \begin{subfigure}[t]{0.24\linewidth}
\centering
      \includegraphics[width=\linewidth]{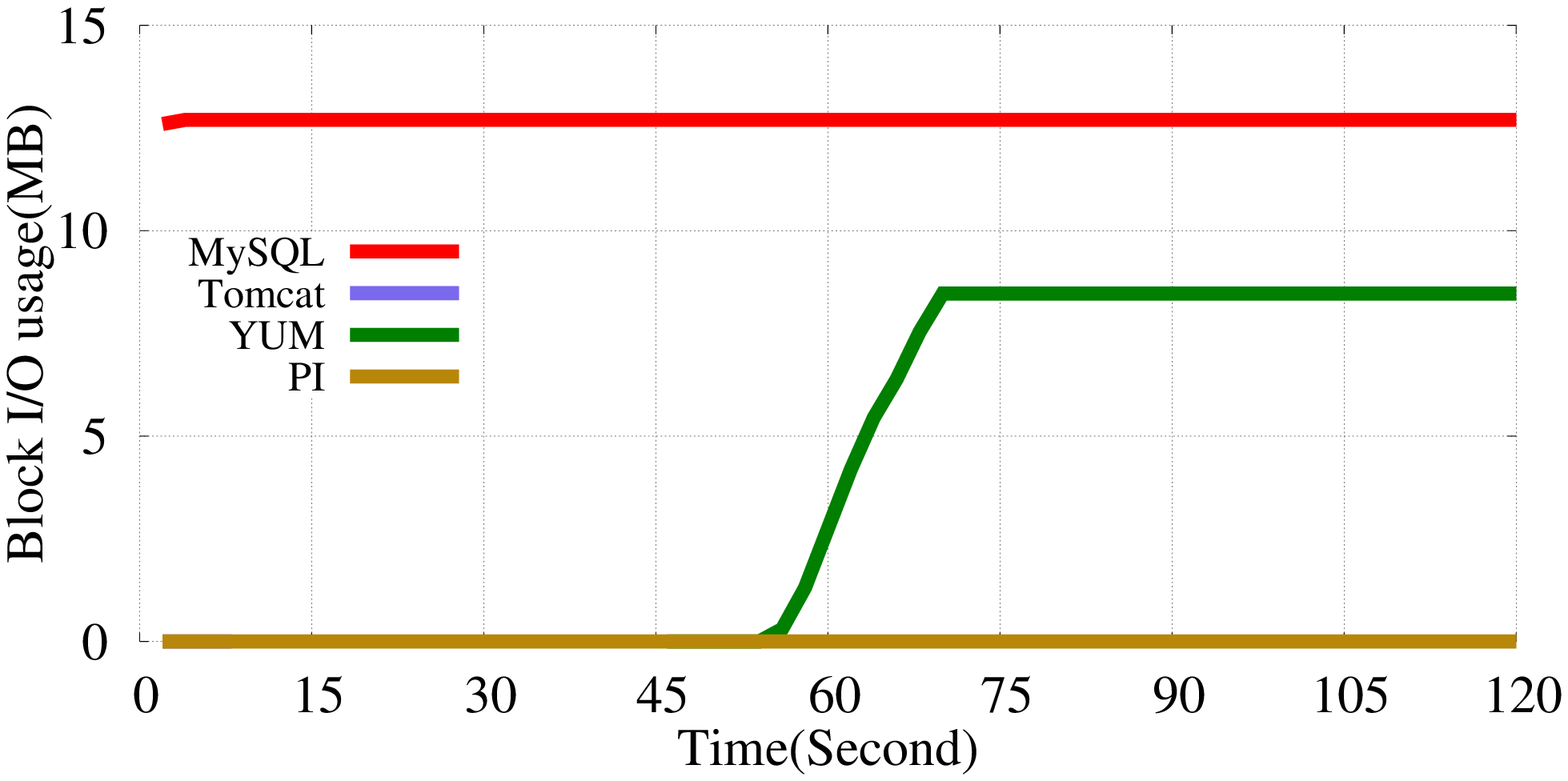}
      \vspace{-0.15in}
      \caption{Usage of Block I/O}
      \label{fig:blockio}
      \end{subfigure} %
\caption{Resource demonds under different workloads on four services, MySQL, Tomcat, YUM, PI.}
\label{fig:understand}
\end{figure*}

\subsection{Identify Resource Demands from Containers}
\label{understand}
Before improving the overall resource efficiency, the system needs to understand the dynamic resource demands of various containers.
A container is, usually, focused on providing a specific service, such as web browsing, data sorting, and database querying.
Different algorithms and operations will be applied to the services, which result in diverse resource demands.
As an intuitive example, we conducted the experiments on NSF Cloudlab~\cite{cloudlab} (M400 node hosted by University of Utah).
The containers are initiated by using the following four images and the data is collected through ``docker stats'' command.
\begin{enumerate}
 \item MySQL: the relational database management system. Tested workloads: scan, select, count, join.
 \item Tomcat: provides HTTP web services with Java. Tested workloads: HTTP queries at 10/second and 20/second of a HelloWorld webpage.
 \item YUM: a software package manager that installs, updates, and removes packages. Tested workload: download and install ``vim'' package.
 \item PI: a service to calculate PI. Tested workload: top 3,000 digits with single thread, top 7,500 digits with two threads.
\end{enumerate}
\cref{fig:cpu,fig:mem,fig:netio,fig:blockio} plots the dynamic resource demands under different workloads on the above four Docker containers.
The figures illustrate very diverse usage patterns on four types of resources: CPU, memory, network I/O, and block I/O.
For example, without workload, container PI consumes very limited resources. However, when the jobs arrive at 10th and 38th second, the CPU usage
jumps to 100\% for a single thread job and 200\% for a two-threads job. The usages of the other three types of resources still remain at very low levels.
For MySQL service container, with tested operations, the CPU usage shows a burst when clients submit a request. At time 84, a ``join'' operation that involves
3 tables is submitted, and we can find CPU usage jumps, as well as memory usage. This is because the join operation needs a lot of computation and copies
of tables in memory. Different usage trends are found on YUM and Tomcat services, where YUM uses less CPU and memory, but more network I/O and block I/O to
download and install packages. On the other hand,
Tomcat consumes a very small amount of network I/O and block I/O due to the size of a tested HelloWorld page, but more than 200MB of memory is used to maintain the
service.
To balance the resource usage, it's crucial to place the containers with complementary demands on the same worker.
As shown on the graphs,
there is a dominant resource demand of a service in a given period despite multiple types of resources.

In \sol~, we need to identify the dominant resource demand for each service.
A manager, in the system, can monitor all of the containers' resource usage and group them by their associated service ID.
Suppose the service $s_i \in S$ contains
$m$ running containers that store in a set, $RC_{s_i}$. The resources consumed by $c_i \in RC_{s_i}$ is denoted by a vector,
$R_{c_i}$, where each attribute, $r_i$, in the vector represents a type of resources, such as memory and CPU.
If there are $q$ types of resources in the system, the average resource cost of $s_i$ is a vector, $R_{s_i}$,
$\begin{aligned}
 R_{s_i} & = \sum\nolimits_{c_i \in RC_{s_i}} R_{c_i} \\
         & = <\sum_{c_i \in RC_{s_i}} r_1 / m, \sum_{c_i \in RC_{s_i}} r_2 / m, ..., \sum_{c_i \in RC_{s_i}} r_q / m>
 \end{aligned}
$

On the worker nodes, there is a limited amount of resources in each type.
The resource limit is a vector that contains $q$ attributes, $<l_1, l_2, ..., l_q>$.
The limit of a system, $<L_1, L_2, ..., L_q>$, is obtained by from the sum of vectors from workers.
Therefore,  $R_{s_i}$ can be represented by a percentage of the total resources in the system, for the
$i^{th}$ type, the container cost for $s_i$ in on average is $\sum_{c_i \in RC_{s_i}} r_i / m \div L_i$.
With the analysis, we define the dominant function,
$
  DOM(s_i) = max \{ \sum_{c_i \in RC_{s_i}} r_i / m \div L_i \}
$
Function $DOM(s_i)$ returns the type of a dominant resource demand of service $s_i$ within a given time period.
The value of $DOM(s_i)$ changes along with the
system depending on the running containers for $s_i$ and the current cost of them.

\subsection{Initial Container Placement}
To use a SwarmKit cluster, clients need to execute the command ``docker run'' to start a new container.
Therefore, the first task for the cluster is to choose a worker node to host the container.
As discussed in section~\ref{back}, the default container placement strategy fails to take dynamic
resource contention into consideration. This is because the managers in SwarmKit
do not have a mechanism that can monitor the current available resource.
\sol, on the other hand, addresses the problem by introducing \sol-Heartbeat.
\sol-Heartbeat is an enhanced heartbeat message that not only the states of worker node, but also the
containers' resource usage over a given time window, the usage includes memory, CPU, bandwidth, and block I/O.
On the manager side, the data will be organized into a table that keeps tracking the current available resource
on each worker and its corresponding containers' resource usages.

Running on managers, Algorithm~\ref{alg:initial} assigns a container initialization task to a specific worker.
Firstly, each manager maintains a known service set that records dockers' characteristics, such as the usage of memory, CPU, bandwidth, and block i/o (line 1).
The initial candidate worker are all running workers (line 2).
When a new container starting task arrives, the algorithm applies all filters that the user specified to shrink the candidate work set, $W_{cand}$ (line 3-6).
Then, it checks whether the container belongs to a known service (line 7). If it is, the $S_{dom}$ parameter will be used to store the container's dominant
resource attribute (line 8). In \sol, we consider four types, memory, CPU, bandwidth, and block i/o.
The $W_{cand}$ set will be sorted according to the dominant resource attribute and
return the $W_{id}$ with the highest available resource in $S_{dom}$ type (line 9-10).
If the service cannot be found in $\{KS\}$, $W_{id}$ with the highest available resource on average will be chosen (line 11-13).

\begin{algorithm}[ht]
\begin{algorithmic}[1]
\STATE Maintains a known characteristics service set $\{KS\}$
\STATE $\{W_{cand}\}$ = All running $W_{id}$;
\STATE {\bf Function ContainerPlacement($S_{ID}$)}

\FOR {$w_{id} \in \{W_{cand}\}$}
\IF {$!Filters(w_{id})$}
\STATE Remove $w_{id}$ from $\{W_{cand}\}$
\ENDIF
\ENDFOR

\IF {$S_{ID} \in \{KS\}$}
\STATE $S_{DOM} =  DOM(S_{ID})$
\STATE Sort $W_{cand}$ according to $r_{S_{DOM}}$
\STATE Return $w_{id}$ with highest $r_{S_{DOM}}$

  \ELSE
  \STATE Sort $\sum_{i=0}^{i=q} r_i / m$ for $w_{id} \in W_{cand}$
  \STATE Return $w_{id}$ with highest average available resource

\ENDIF

\end{algorithmic}
\caption{Container Placement on Managers}
\label{alg:initial}
\end{algorithm}

\subsection{Migrating a Container}
In a Swarmkit cluster, resource contention happens on every worker.
The container conitor, which is a module of \sol, runs on each worker to record resource usages of hosting containers.
In addition, the worker keeps tracking available resources on itself. Whenever it finds a draining type of
resources becomes a bottleneck, it sends to managers a \sol~ alert message that contains the bottleneck type and the
most costly container of this type. Upon receiving the \sol~ alert message, the manager needs
to migrate this container to an appropriate worker and kill it on the worker to release the resources.

Algorithm~\ref{alg:migrate} presents the procedure to process an alert message from $w_i$.
It first builds a candidate set $W_{cand}$, which includes all running workers expect $w_i$ that sends the alert (line 1).
Then, the manager extracts the resource type, $r_i$ that causes the bottleneck and finds the corresponding $S_{id}$ for the $C_{id}$ (lines 2-4).
With $W_{cand}$ and $S_{id}$, the algorithm can decide whether this $S_{id}$ is a global service (line 5).
If $S_{id}$ is a global service and it is in the known service set, $\{KS\}$, the algorithm returns $w_{id}$ that is included in $W_{cand}$,
with the highest available $r_{S_{DOM}}$.
On the other hand, it returns $w_{id}$ with the highest available $r_i$ if $S_{id}$ is not in $\{KS\}$ and $S_{DOM}$ on unknown (lines 6-12).
When $S_{id}$ is not a global service, we want to increase the reliability of $S_{id}$ by placing its containers to different workers as much as possible.
In this situation, we have a similar process expect a different $W_{cand}$, where $W_{cand}$  contains all running
workers that do not hosting any containers for $S_{id}$ (lines 13 - 23).

\begin{algorithm}[ht]
\begin{algorithmic}[1]

\STATE $\{W_{cand}\}$ = All running workers expect $w_i$;
\STATE {\bf Function ReceiveAlertMsg($C_{id}$)}
\STATE Extract the bottleneck type $r_i$
\STATE Find corresponding $S_{id}$ for $C_{id}$

\IF {$\forall w_{id} \in W_{cand} \rightarrow S_{id} \in w_{id}$}

  \IF {$S_{id} \in \{KS\}$}
  \STATE $S_{DOM} = DOM(S_{id})$


\STATE Sort $W_{cand}$ according to $r_{S_{DOM}}$
\STATE Return $w_{id}$ with highest $r_{S_{DOM}}$

\ELSE
\STATE Sort $W_{cand}$ according to $r_i$
\STATE Return $w_{id}$ with highest $r_i$

\ENDIF

  \ELSE
  \FOR {$ w_{id} \in W_{cand}$}
  \IF {$S_{id} \in w_{id}$}
  \STATE Remove $w_{id}$ from $W_{cand}$
  \ENDIF
  \ENDFOR

\IF {$S_{id} \in \{KS\}$}
  \STATE $S_{DOM} = DOM(S_{id})$


\STATE Sort $W_{cand}$ according to $r_{S_{DOM}}$
\STATE Return $w_{id}$ with highest $r_{S_{DOM}}$

\ELSE
\STATE Sort $W_{cand}$ according to $r_i$
\STATE Return $w_{id}$ with highest $r_i$

\ENDIF

\ENDIF

\end{algorithmic}
\caption{Process \sol~ Alert Message from $w_i$}
\label{alg:migrate}
\end{algorithm}

\section{Performance Evaluation}

\begin{figure*}[ht]
   \centering
      \begin{subfigure}[t]{0.32\linewidth}
\centering
      \includegraphics[width=\linewidth]{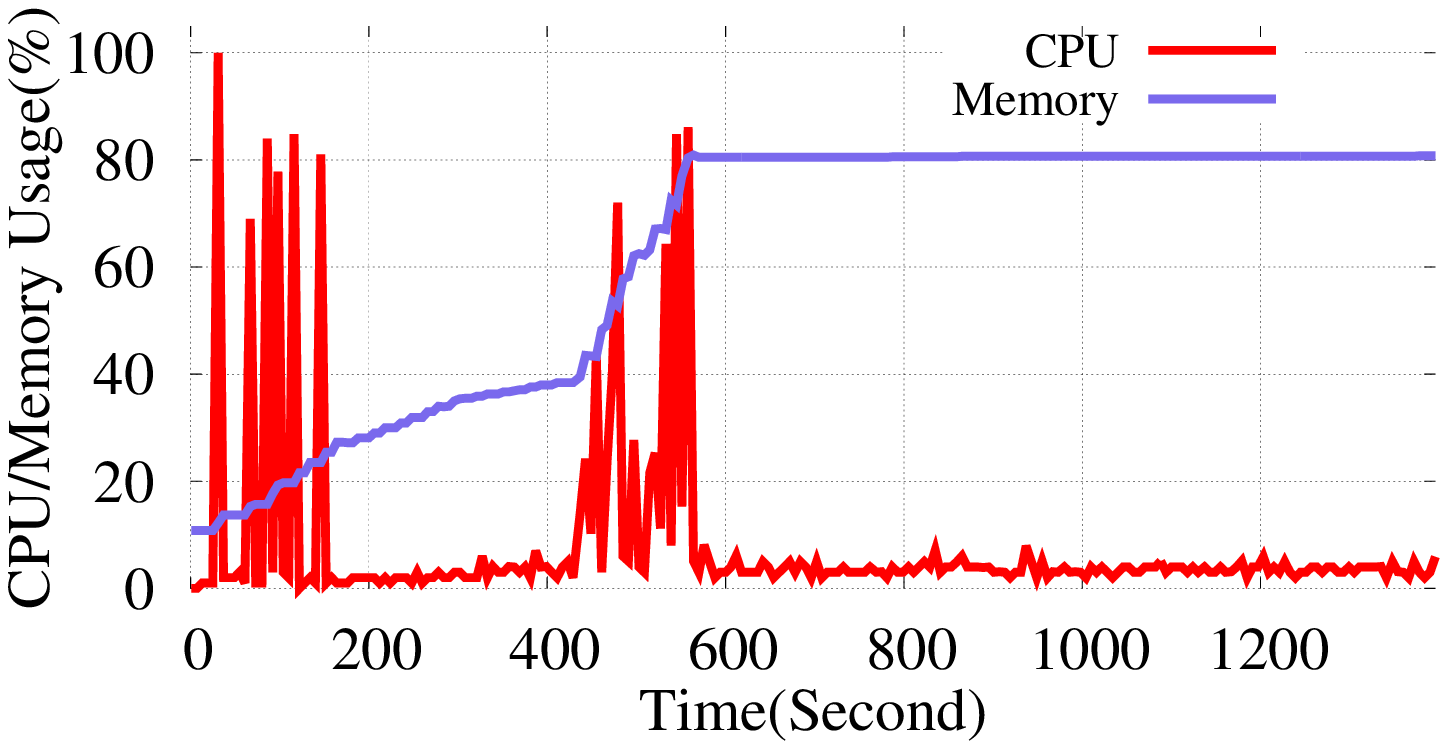}
      \vspace{-0.15in}
      \caption{Worker 1 (Spread, 100)}
      \label{fig:25worker1}
      \end{subfigure} %
      ~
      \begin{subfigure}[t]{0.32\linewidth}
\centering
      \includegraphics[width=\linewidth]{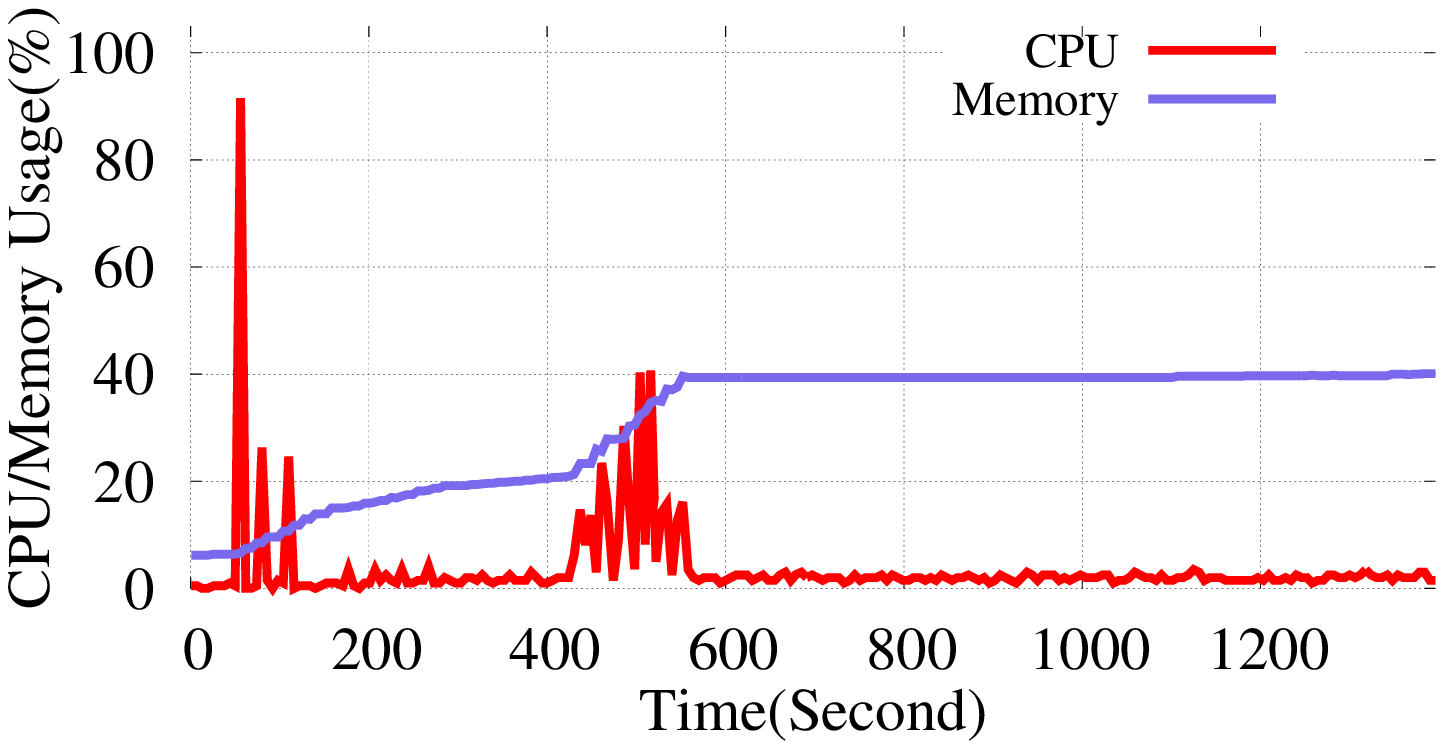}
      \vspace{-0.15in}
      \caption{Worker 2 (Spread, 100)}
      \label{fig:25worker2}
      \end{subfigure} %
      ~
      \begin{subfigure}[t]{0.32\linewidth}
\centering
      \includegraphics[width=\linewidth]{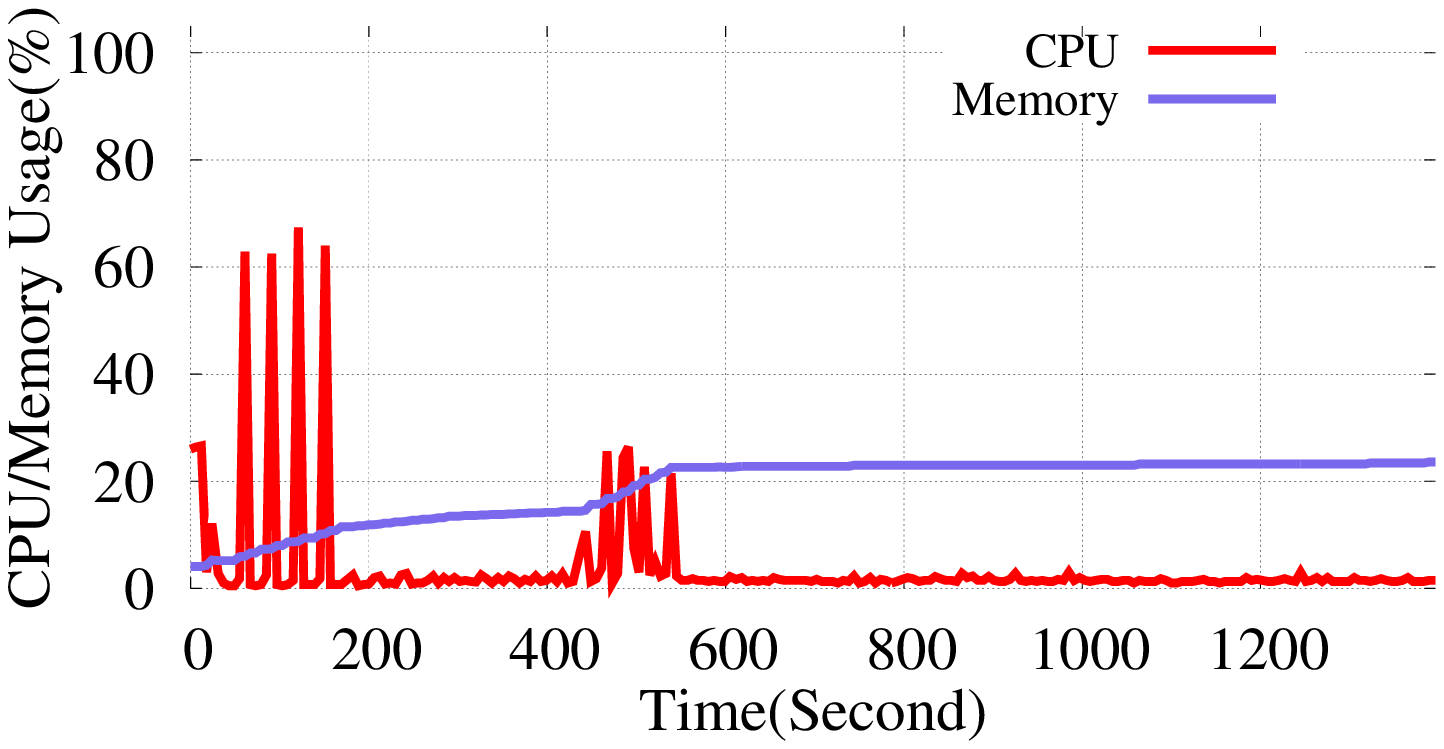}
      \vspace{-0.15in}
      \caption{Worker 3 (Spread, 100)}
      \label{fig:25worker3}
      \end{subfigure} %
      
   \centering
      \begin{subfigure}[t]{0.32\linewidth}
\centering
      \includegraphics[width=\linewidth]{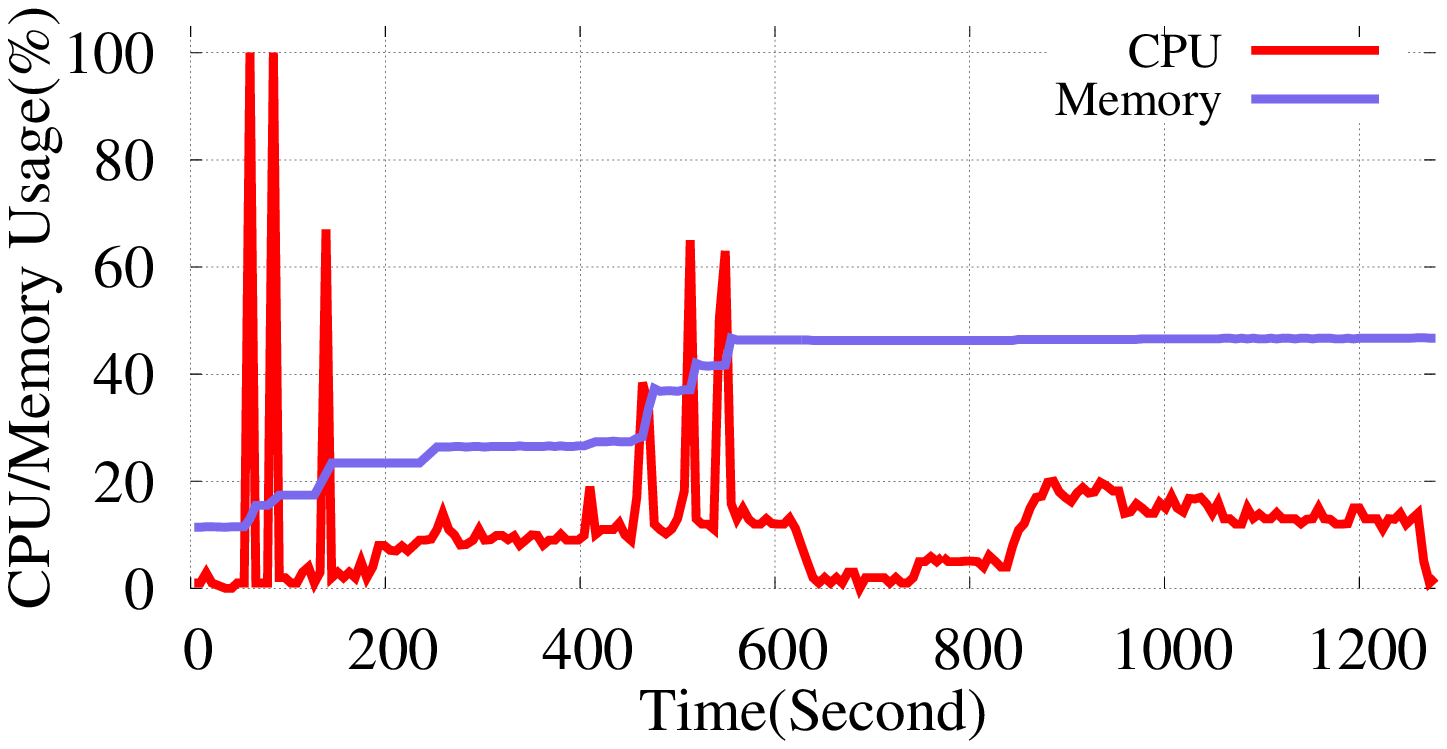}
      \vspace{-0.15in}
      \caption{Worker 1 (\sol, 100)}
      \label{fig:25worker1g}
      \end{subfigure} %
      ~
      \begin{subfigure}[t]{0.32\linewidth}
\centering
      \includegraphics[width=\linewidth]{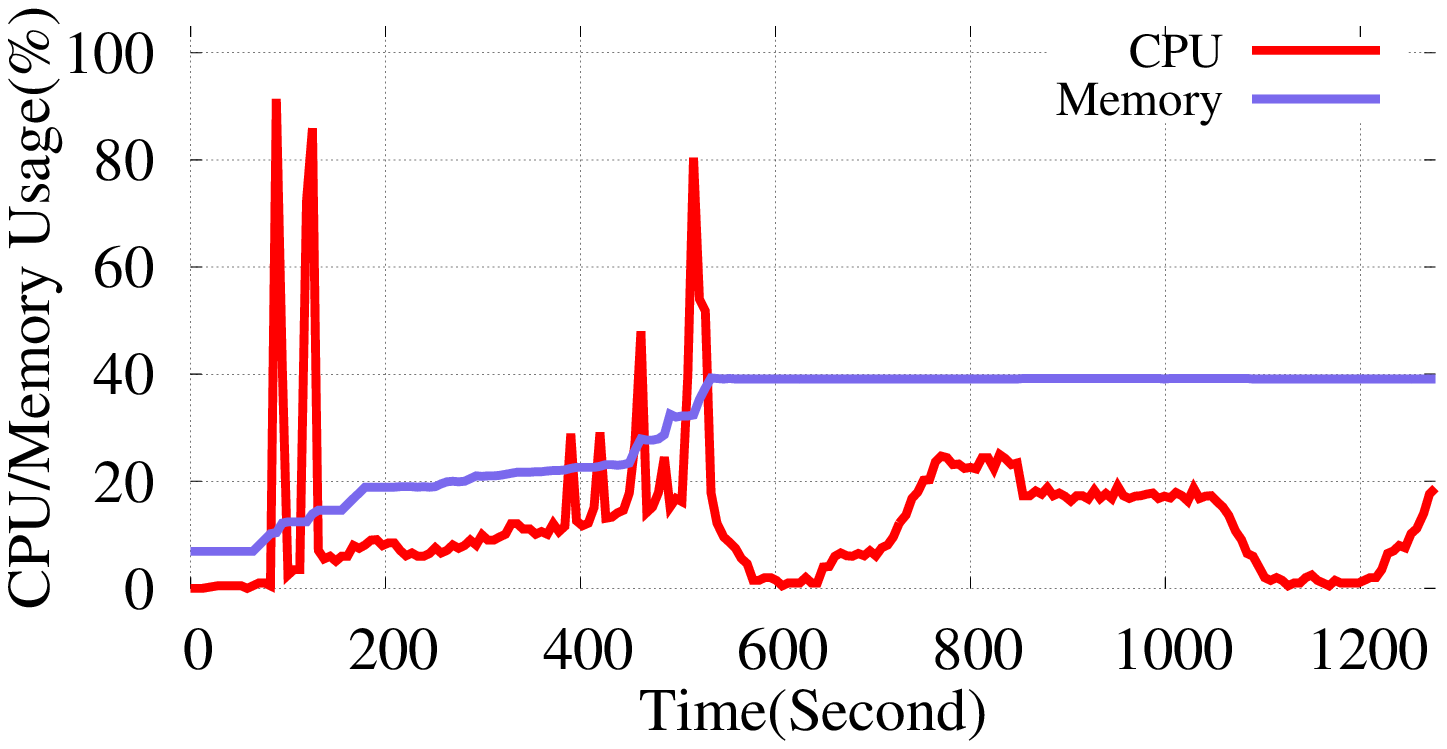}
      \vspace{-0.15in}
      \caption{Worker 2 (\sol, 100)}
      \label{fig:25worker2g}
      \end{subfigure} %
      ~
      \begin{subfigure}[t]{0.32\linewidth}
\centering
      \includegraphics[width=\linewidth]{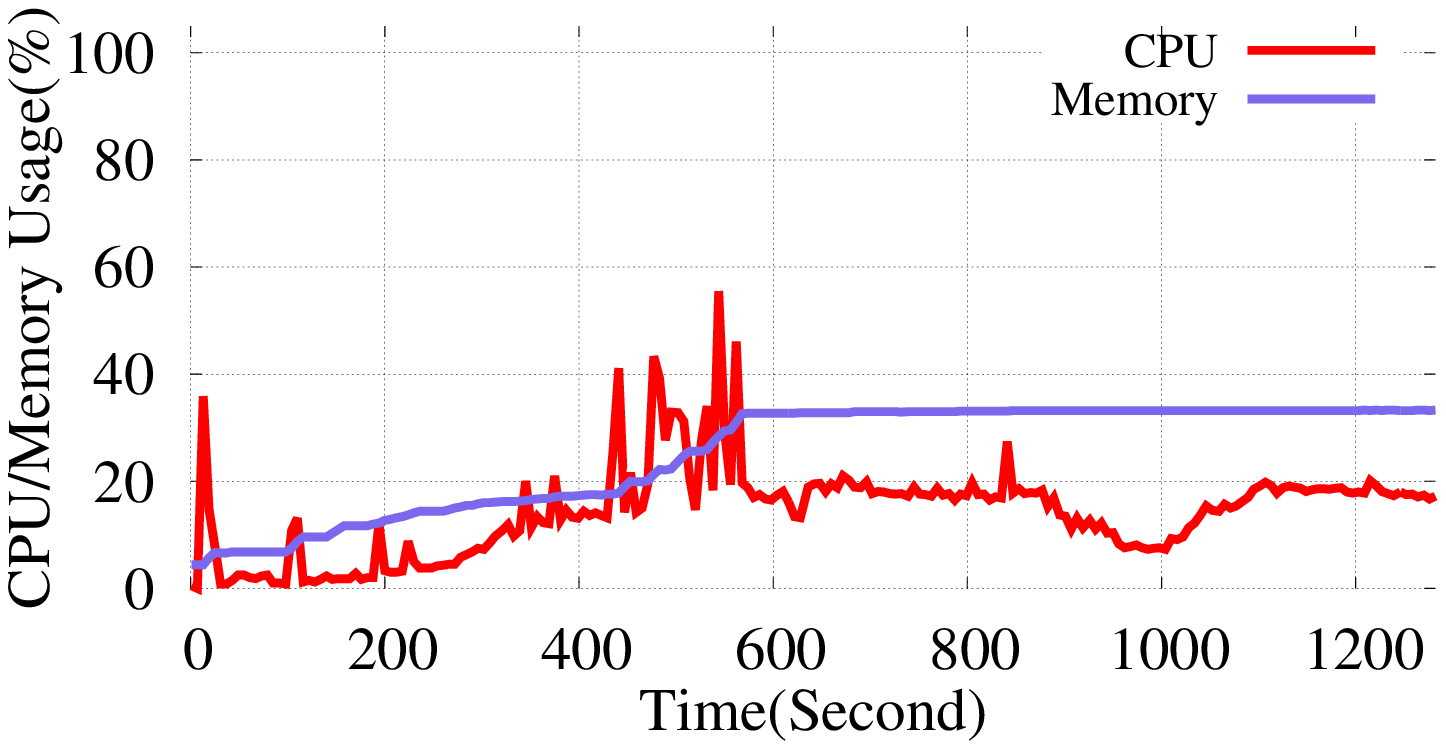}
      \vspace{-0.15in}
      \caption{Worker 3 (\sol, 100)}
      \label{fig:25worker3g}
      \end{subfigure} %
      
\caption{Memory and CPU resources usage comparison between Spread and \sol~placement scheme (100 containers)}
\label{fig:noworkload25}
\end{figure*}

\subsection{Implementation, Testbed and Workloads}
We implement our new container placement 
scheme, \sol, on Docker Community Edition (CE) v17. 
As described in Section~\ref{sys}, the major modules in \sol~are integrated into the existing
Docker Swarmkit framework. 

To evaluate \sol, we build a heterogeneous cluster on Alibaba Cloud~\cite{aliyun}, which supports multiple types of computing nodes.
Specifically, we use three different types of instances, small (1 CPU core and 4G memory), medium (4 CPU cores and 8G memory)
and large (8 CPU cores and 16G memory).
In the small-scale testing, we setup a cluster with 3 nodes, one of each type, and configure it with 1 manager and 3 worker (1 of the 3 physical nodes 
hosts both manager and worker). In experiments on scalability, we configure the cluster with 1 manger and 9 workers that consist 3 instances of each type. 

The main objective of \sol~is to understand resource demands of services and place them on appropriate worker nodes.
As we discussed in Section~\ref{understand}, characteristics of services are varied. 
Therefore, workloads for the cluster are images of various services. In the evaluation,
we select 18 different images in 4 types from Docker Hub~\cite{dockerhub} to build our image pool. 
{\bf Database Services}:  MongoDB, MySQL, Postgres, Cassandra, RethinkDB.
{\bf Storage/Caching Services}:  Registry, Memcached.
{\bf Web Services}: Tomcat, Httpd, Redis, HAProxy, Jetty, Nginx, GlassFish.
{\bf Message Services}: RabbitMQ, Apache ZooKeeper, ActiveMQ, Ghost.

\subsection{Evaluation Results}
\subsubsection{Idle containers}
In this subsection, we present the result of a cluster with idle containers.
If a container is in a running state but does not serve any clients, we call it a idle container.
Idle container is an important concept since every node, right after initialization will act as an idle container.
Understanding the resource demands of an idle container will help us select its host.
In these experiments, we first randomly choose 14 images form the pool, and each image will be used to initiate 10
containers. Therefore, there are 140 containers in the cluster. 
Those containers are started one by one with 5 seconds interval. 
This is because previous containers will result in  different available resources on worker nodes, which we can utilize 
to test \sol.

Fig~\ref{fig:noworkload25} illustrates a comparison of memory and CPU usages between Spread, 
a Swarmkit default placement scheme, with \sol. As we can see from the subfigures, most of the CPU usage happens from 0 to 500s.
This is caused by submission pattern that used to initiate containers. The percentage grows continuously from 0 to 500s since
we have 100 containers and the submission interval is 5 seconds. While in both systems, the usage of CPU stays at a low level on average. 
However, the memory usage keeps increasing along with the number of containers on each worker. 
Due to the idle container setting, the utilization of memory is stable after 500s (all the containers 
have successfully initiated). There are some jitters on the curve of CPU, because 
some supporting programs, such as the Docker engine and service daemon, are running on the same worker and, of course, they are consuming resources. 
Comparing the memory usage rates after 500s, \sol~significantly reduces rate on worker 1, from 80.5\% to 46.7\%. 
On worker 2, Spread and \sol~achieve similar performance on memory, 39.1\% verse 40.6\%.
On worker 3, Spread results in 23.6\% and \sol~consumes 33.3\%. The \sol~outperforms Spread by considering the heterogeneity in the cluster and 
various resource demands of services. When a task arrives at the system, it selects a worker based on the service demands and current available resources.
Fig~\ref{fig:num} shows the number of containers on workers. For Swarmkit with Spread, it uses a bin-pack strategy and tries to equally distribute
the containers to every worker, which results in 34, 33, 33 containers for worker 1, 2, 3. While in \sol, worker 3 has more power than others and hosts more 
containers than worker 1, which has limited resource.

\begin{figure}[ht]
   \centering
      \begin{minipage}[t]{0.48\linewidth}
\centering
      \includegraphics[width=\linewidth]{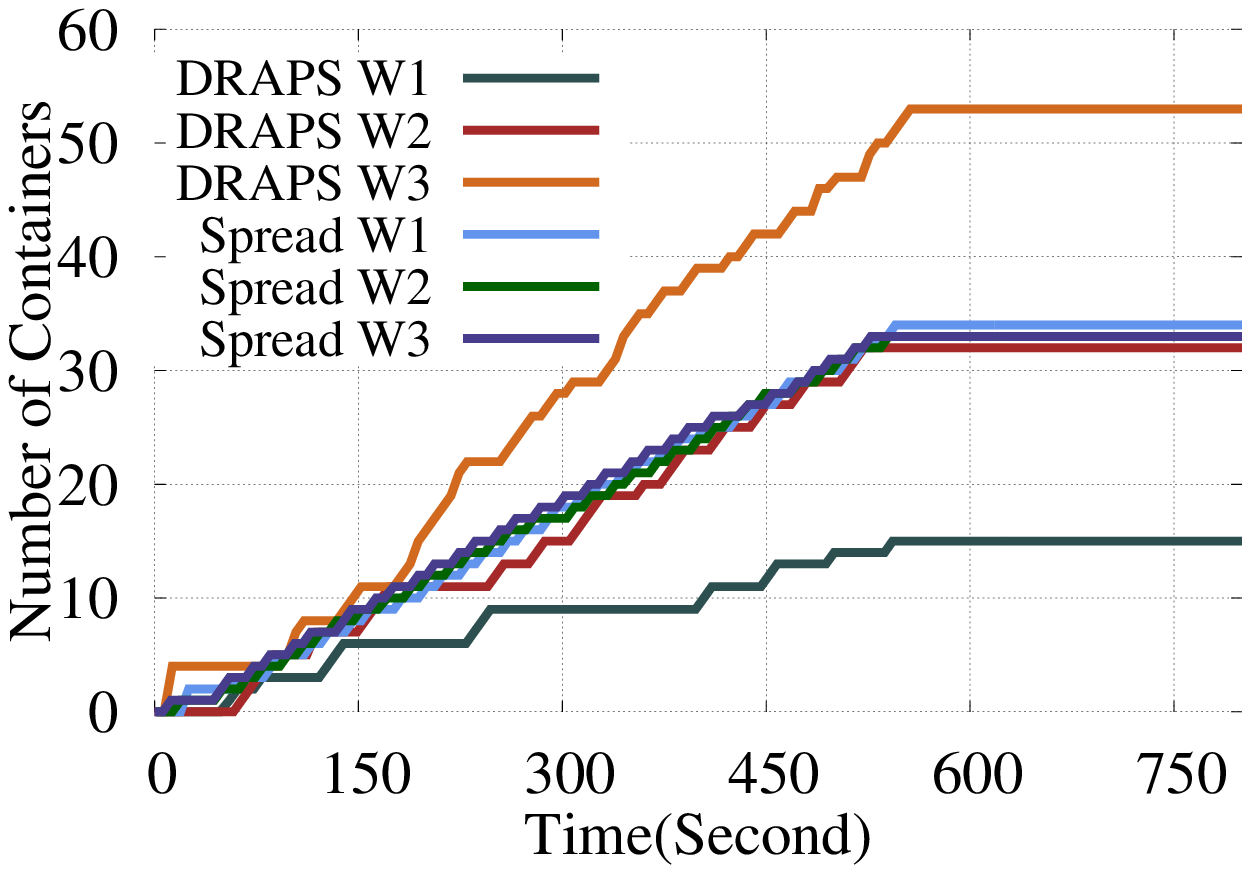}
      \caption{Number of containers on each worker}
      \label{fig:num} 
      \end{minipage} %
      \begin{minipage}[t]{0.48\linewidth}
\centering
      \includegraphics[width=\linewidth]{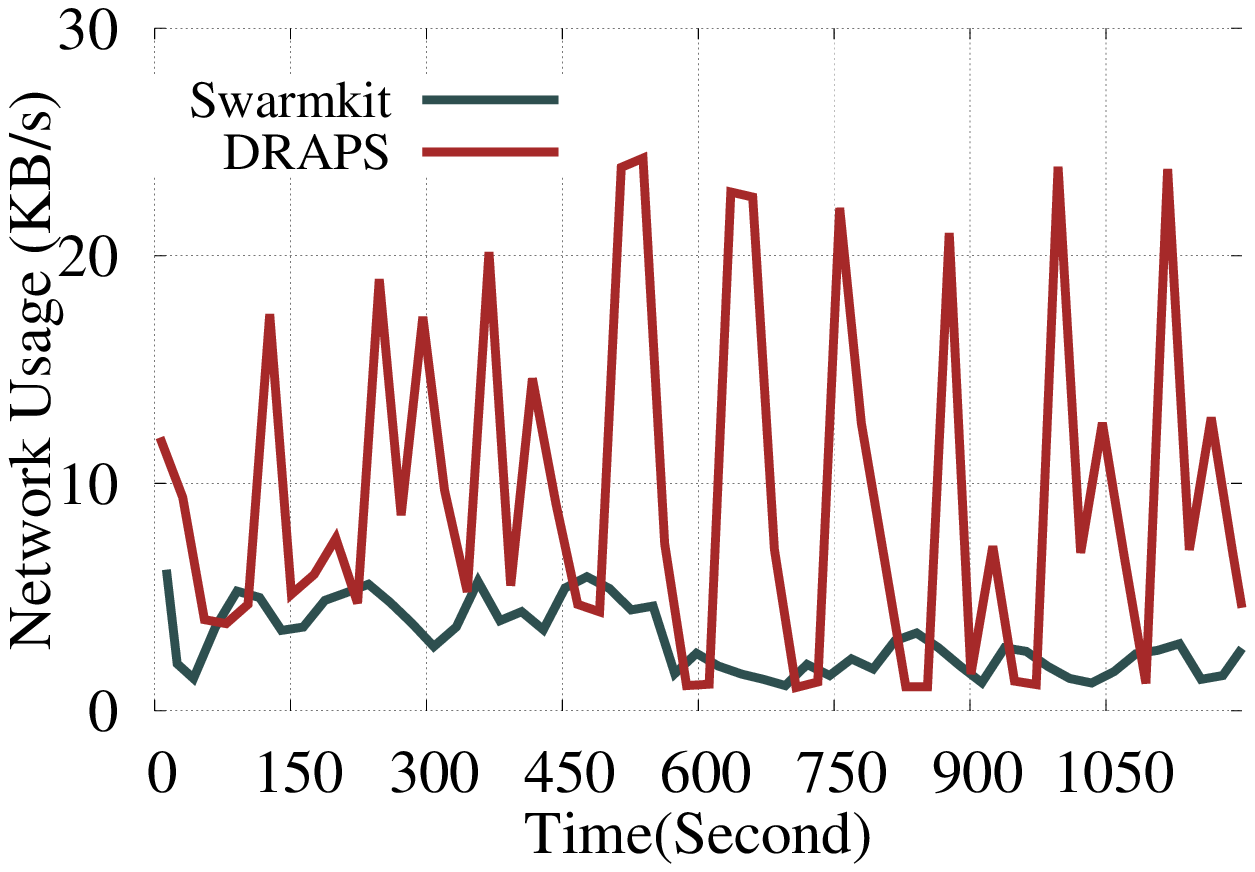}
      \caption{Network consumption comparison on worker 3}
      \label{fig:25net}
      \end{minipage} %
\end{figure}

While \sol~achieves better performance, it introduces more data transfers between managers and workers through heartbeat messages.
Fig~\ref{fig:25net} plots the network consumption of Swarmkit and \sol~ on worker 3, which hosts both a manager and a worker. 
As expected, \sol~ consumes more bandwidth than Swarmkit due to the enhanced heartbeat messages includes more statistical information resource
usages of containers. Considering the distributed architecture, the system can have multiple managers and each of them in charge of controllable
number of workers, the increase of bandwidth consumption that brought by \sol~ is reasonable.

\begin{figure*}[ht]
   \centering
      \begin{subfigure}[t]{0.32\linewidth}
\centering
      \includegraphics[width=\linewidth]{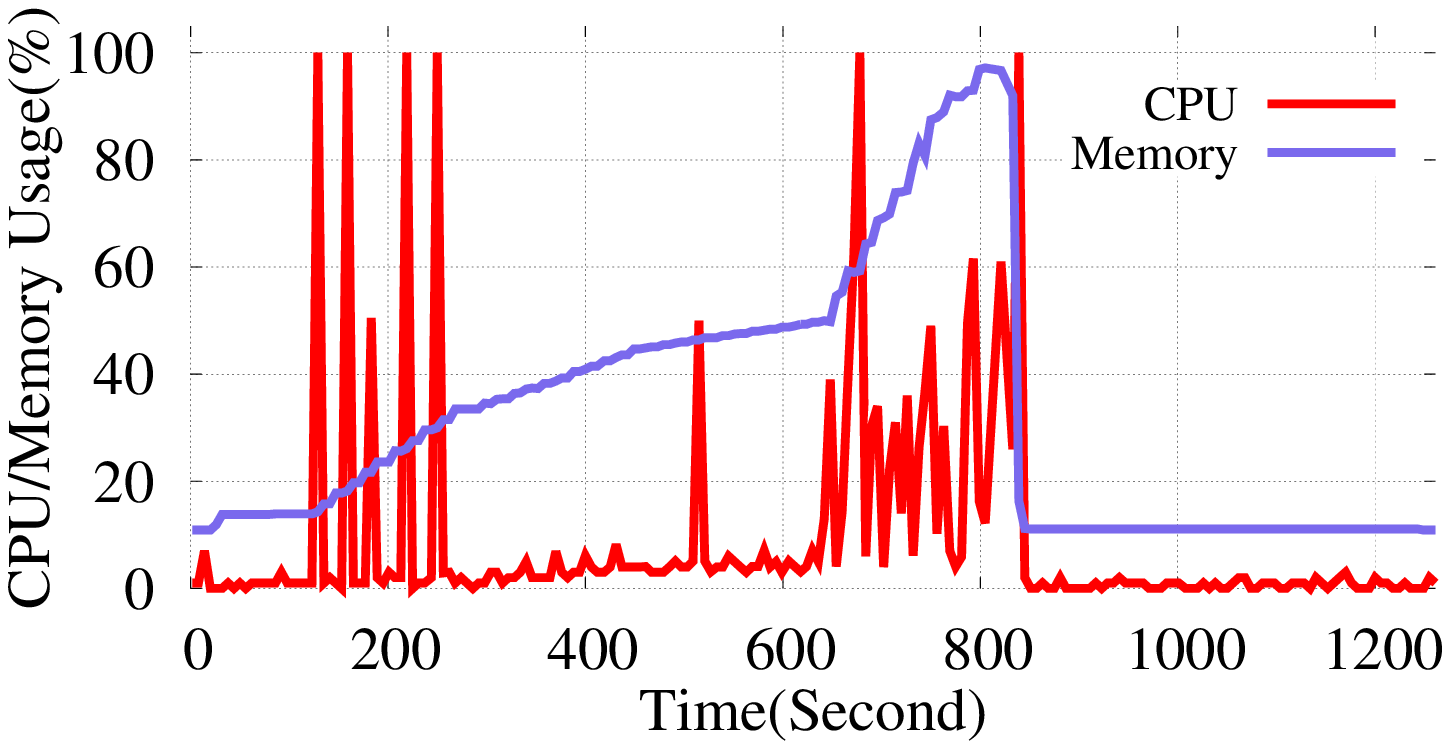}
      \vspace{-0.15in}
      \caption{Worker 1 (Spread, 140)}
      \label{fig:worker1}
      \end{subfigure} %
      ~
      \begin{subfigure}[t]{0.32\linewidth}
\centering
      \includegraphics[width=\linewidth]{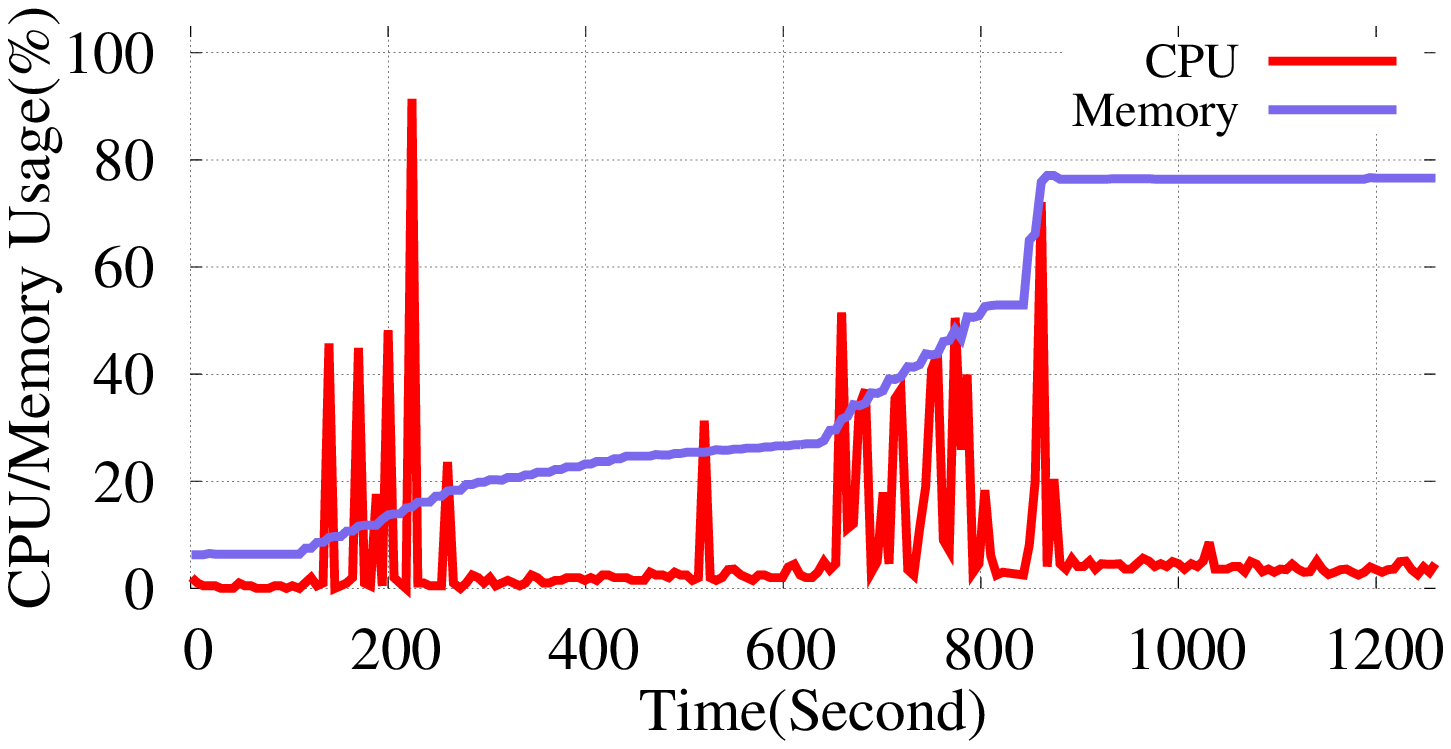}
      \vspace{-0.15in}
      \caption{Worker 2 (Spread, 140)}
      \label{fig:worker2}
      \end{subfigure} %
      ~
      \begin{subfigure}[t]{0.32\linewidth}
\centering
      \includegraphics[width=\linewidth]{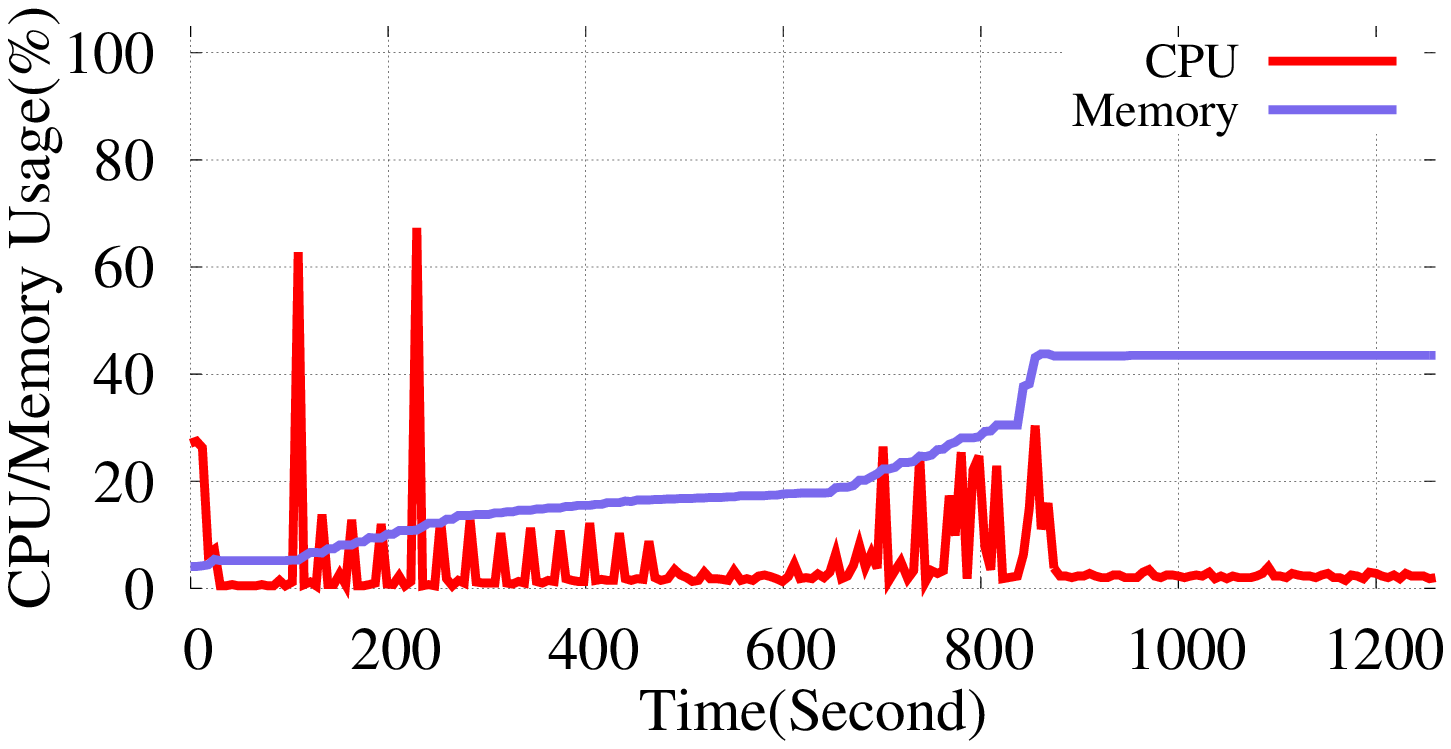}
      \vspace{-0.15in}
      \caption{Worker 3 (Spread, 140)}
      \label{fig:worker3}
      \end{subfigure} %
      
   \centering
      \begin{subfigure}[t]{0.32\linewidth}
\centering
      \includegraphics[width=\linewidth]{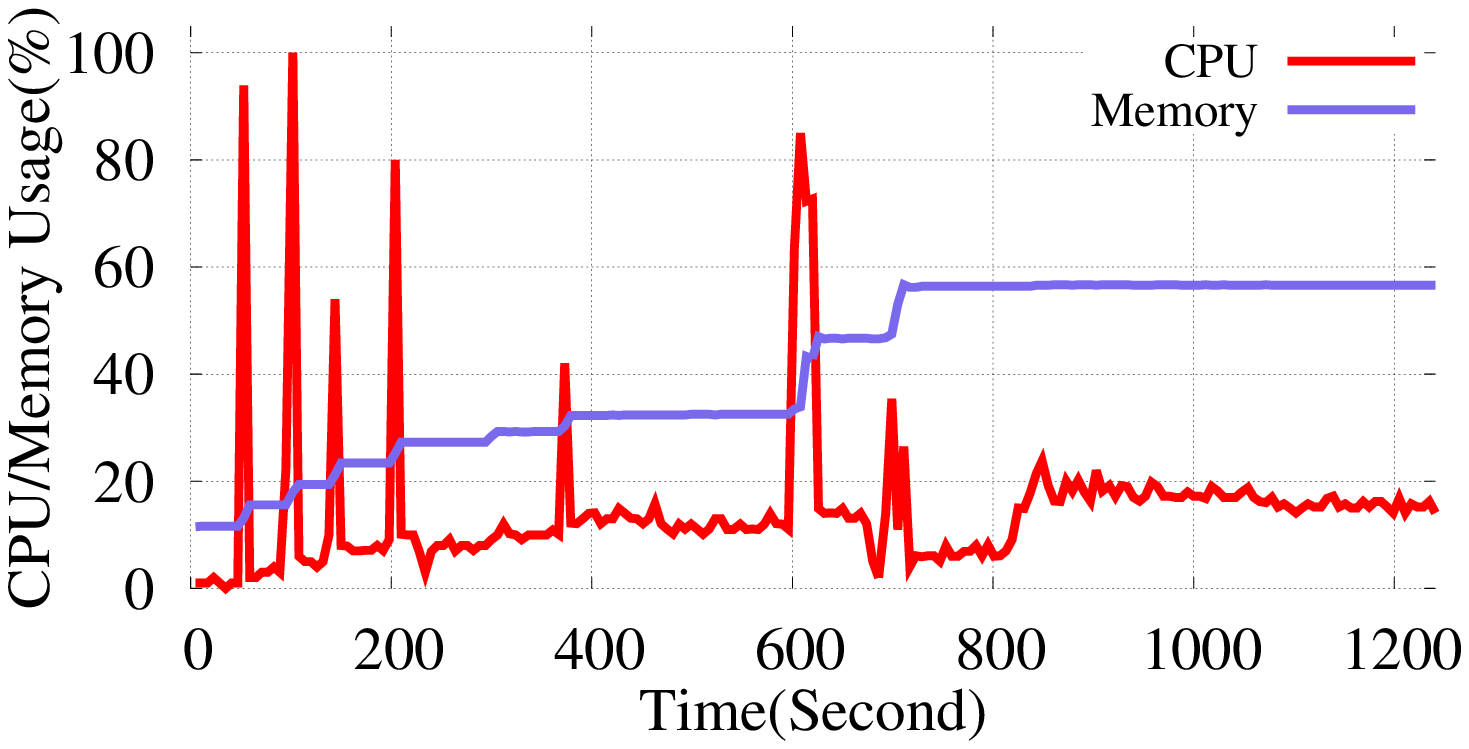}
      \vspace{-0.15in}
      \caption{Worker 1 (\sol, 140)}
      \label{fig:worker1g}
      \end{subfigure} %
      ~
      \begin{subfigure}[t]{0.32\linewidth}
\centering
      \includegraphics[width=\linewidth]{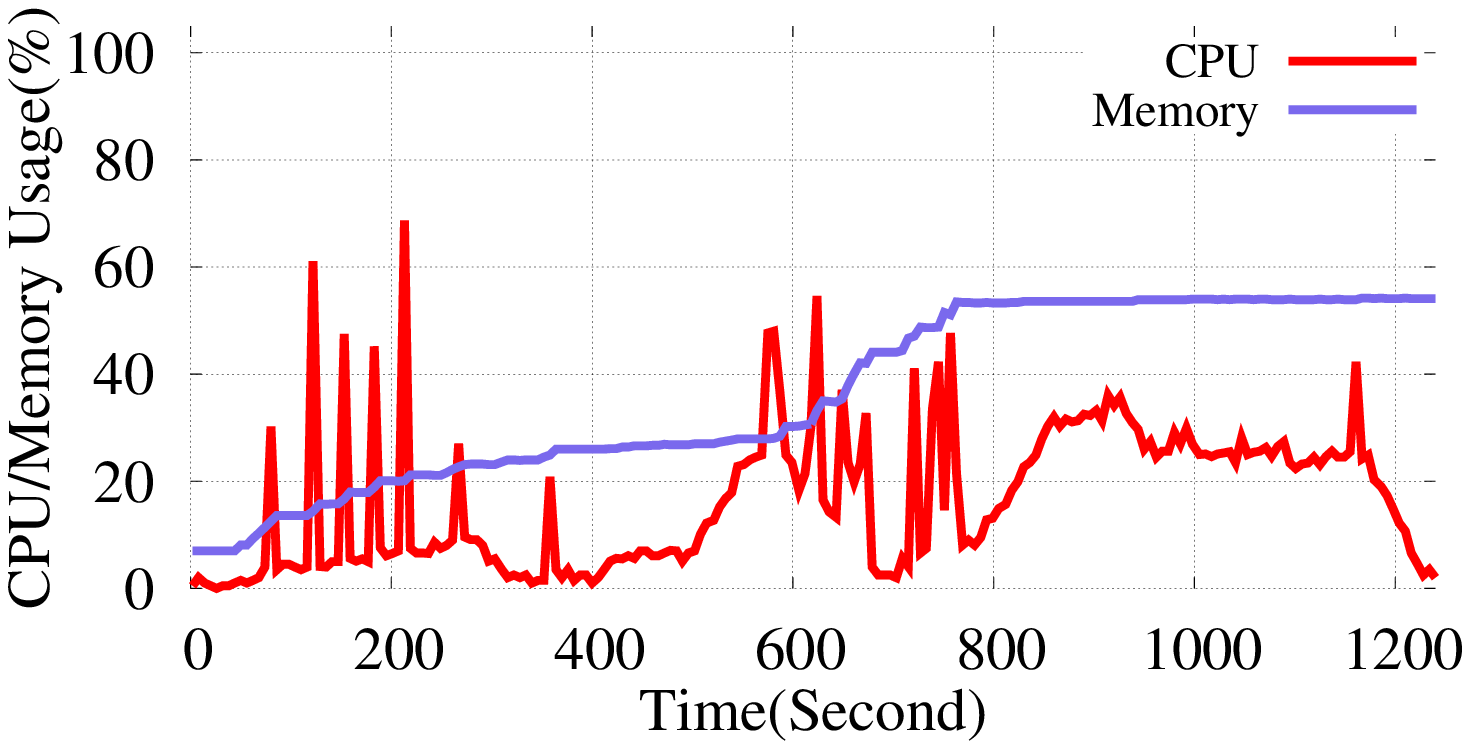}
      \vspace{-0.15in}
      \caption{Worker 2 (\sol, 140)}
      \label{fig:worker2g}
      \end{subfigure} %
      ~
      \begin{subfigure}[t]{0.32\linewidth}
\centering
      \includegraphics[width=\linewidth]{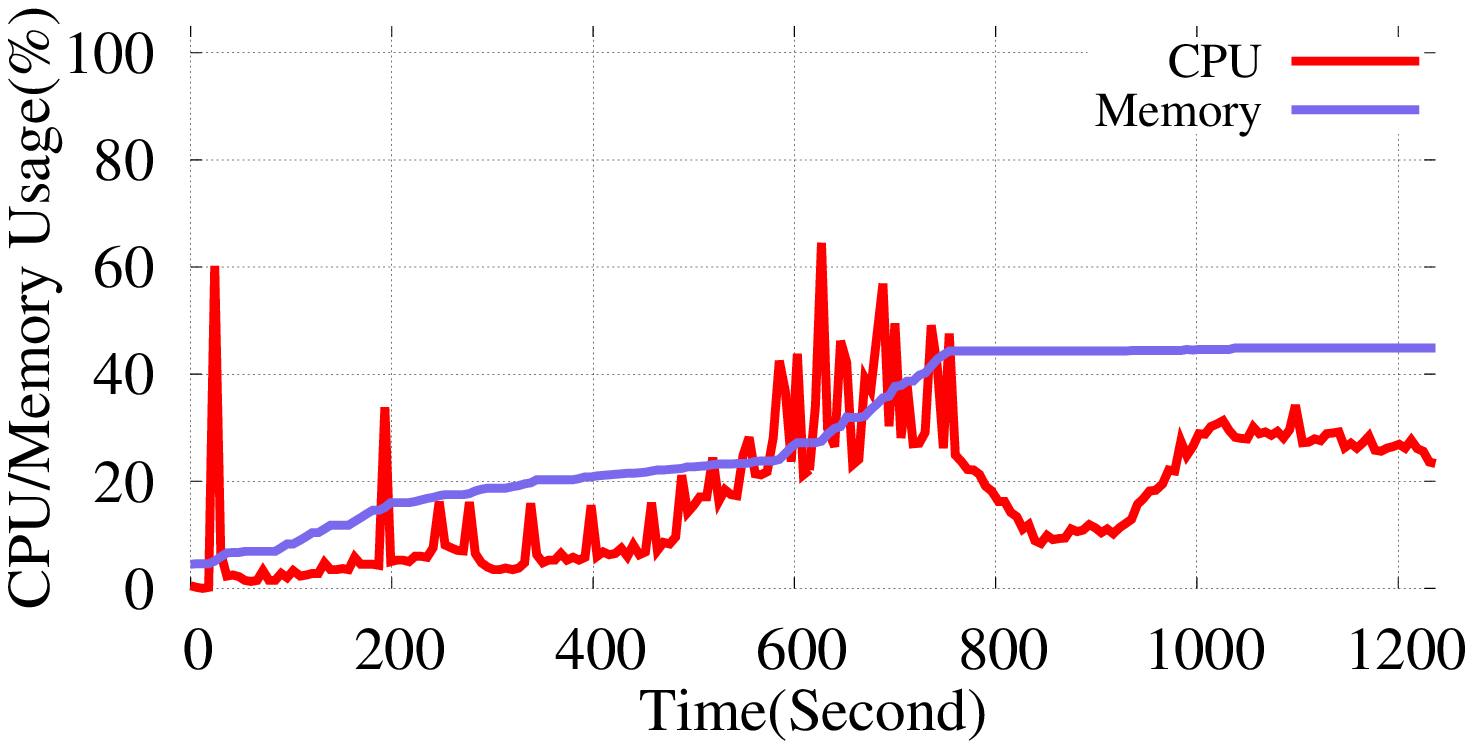}
      \vspace{-0.15in}
      \caption{Worker 3 (\sol, 140)}
      \label{fig:worker3g}
      \end{subfigure} %
      
\caption{Memory and CPU resources usage comparison between Spread and \sol~placement scheme (140 containers)}
\label{fig:noworkload35}
\end{figure*}

Next, we conduct the same experiments with 40\% more containers to test the scalability of \sol.
Fig~\ref{fig:noworkload35} plots the system performance with 140 Docker containers. Comparing the figures, the first impression is that
on Fig~\ref{fig:worker1}, the usages suddenly drop from 95.2\% to 11.1\% for memory and 100\% to 0 for CPU.
The reason lies in the fact that, at time 726, the memory becomes bottlenecked on work 1 with Spread scheme.
However, the manager does not award this situation on worker 1, and assign a new container to it. Worker 1 fails to 
start the new container, and drains the memory, which results in the death of all containers on it. The Docker engine
decides to kill them all when it can not communicate with them. 
On the other hand, \sol~considers dynamic resources usages on workers, and it stops assigning task to a worker if it has already overwhelming.
It is shown on Fig~\ref{fig:worker1g} that the usages of memory and CPU remains at 46.3\% and 18.8\% for worker 1 with \sol.
While worker 2 with Spread still runs smoothly at the end of the testing, its memory usage is at a high level, 76.6\%, comparing to work 2 with \sol~
the value is 54.1\%.

\subsubsection{Loaded containers}
Besides idle containers, we set up a mix environment that includes both idle and loaded containers. 
If clients are generating workloads to the services on the running containers, we call it loaded containers.
Evidenced by Fig.~\ref{fig:understand}, we know that loaded containers consume more resources than idle ones. In addition,
the usage pattern of a loaded container changes along with the workload.
Fig~\ref{fig:120workload} plots the memory usage and number of containers on Worker-1. 
For the experiments running with Spread, it drains the memory at time 825s that the memory usage drops from 98.5\% to 11.9\%.
Simultaneously, the number of running containers on worker-1 drops from 44 to 9 and then, to 0 at time 825s and 837s.
This is because the docker engine kills all containers when the memory is not enough to maintain the system itself.
Due to less containers on worker-1 with \sol~(44 v.s 24), it runs normally throughout the entire experiments.
Fig~\ref{fig:iowait} shows the value of I/O wait in percentage, which measures the percent of time the CPU is idle, but waiting for an I/O to complete.
It shows a similar trend that at time 849s the value drops to 0 for Spread, while \sol~maintains stable performance.
\begin{figure}[ht]
   \centering
      \begin{minipage}[t]{0.48\linewidth}
\centering
      \includegraphics[width=\linewidth]{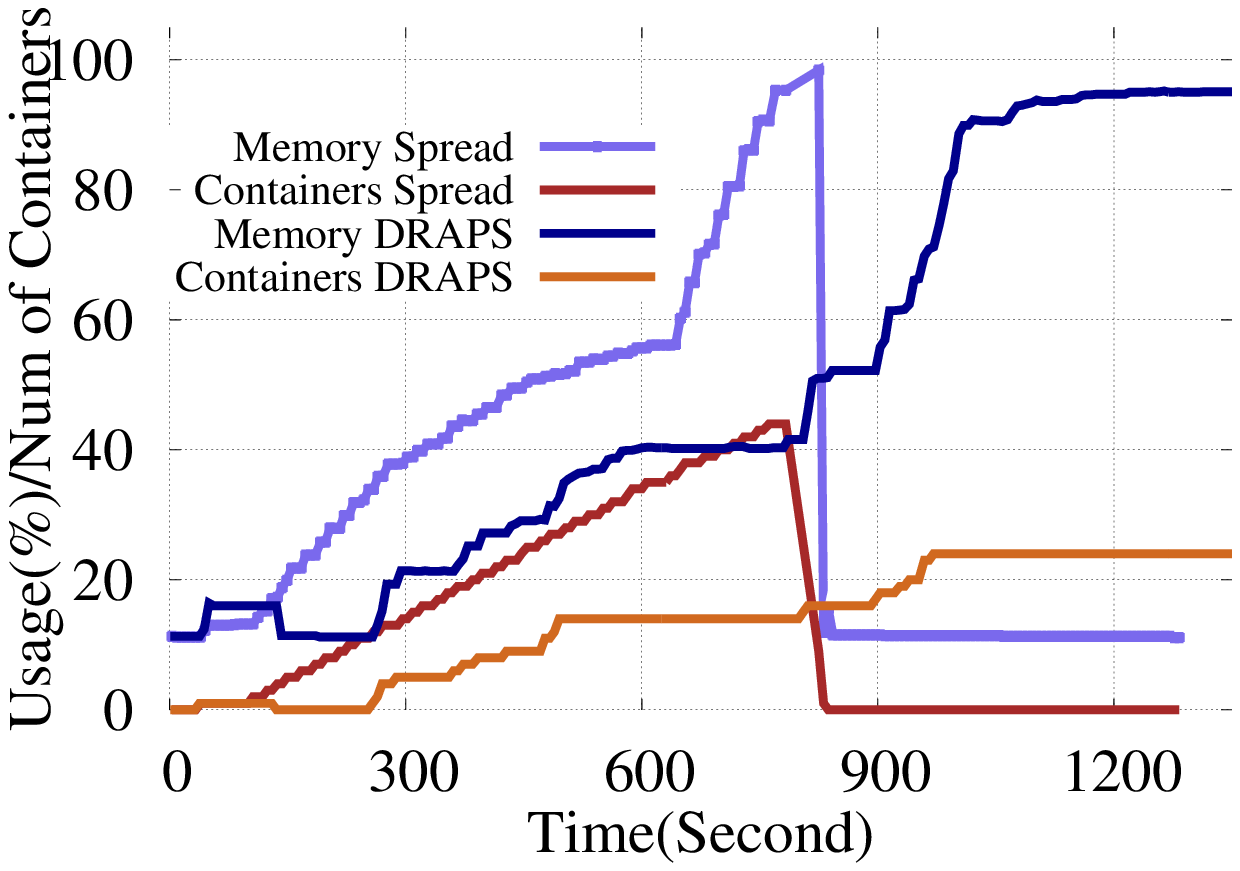}
      \caption{Memory usage and container number on worker1}
      \label{fig:120workload} 
      \end{minipage} %
      \begin{minipage}[t]{0.48\linewidth}
\centering
      \includegraphics[width=\linewidth]{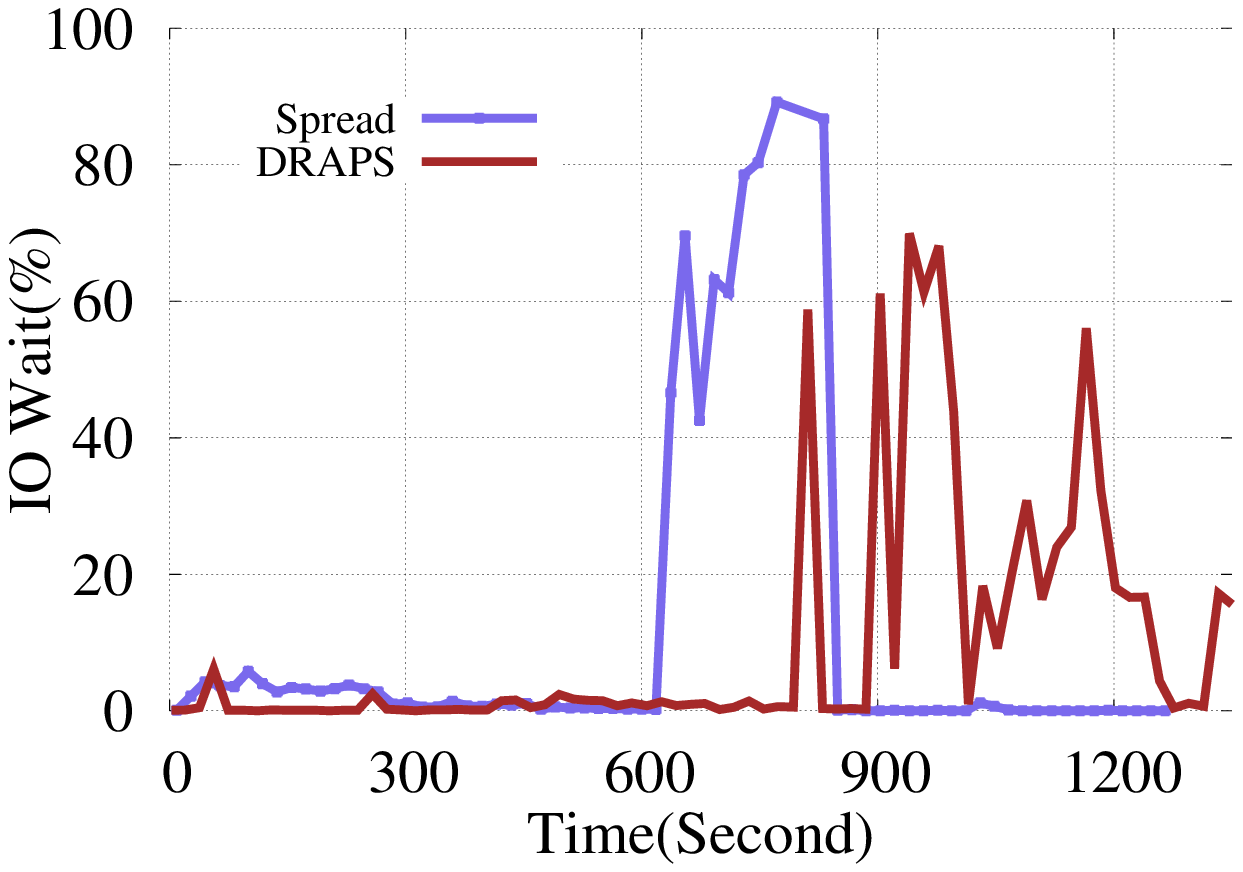}
      \caption{Value of I/O Wait on worker1}
      \label{fig:iowait}
      \end{minipage} %
\end{figure}

\section{Conclusion}
This paper studies the container placement strategy in a heterogeneous cluster.
We target on distributing containers to the worker nodes with the best available resources.
In this paper, we develop \sol, which considers various resource demands from containers and 
current available resources on each node. We implemented \sol~on Docker Swarmkit platform and 
conducted extensive experiments. 
The results show a significant improvement on the 
system stability and scalability when comparing with the default Spread strategy.

\bibliographystyle{myunsrt}
\bibliography{routing}

\end{document}